\documentclass[12pt,final]{article}
\usepackage{geometry} 
\usepackage{gt-style-private}
\usepackage{color}
\usepackage{amssymb,amsmath,amssymb,amsthm}
\usepackage{hyperref} 
\usepackage{bm} 
\usepackage{todonotes}
\usepackage{framed}
\usepackage{mathrsfs}  
\newcommand{\mynote}[1]{}

\usepackage{pdfpages} 

\definecolor{gt}{RGB}{0,0,80}
\definecolor{appunti}{RGB}{0,255,255}

\newcommand\blue{\color{black}}
\newcommand\black{\color{black}}

\title{Linear models for thin plates of polymer gels.}
\author{R. Paroni\footnote{DADU, University of Sassari, Palazzo del Pou Salit, 07041 Alghero (SS), Italy. Email:\texttt{paroni@uniss.it}} \ and G. Tomassetti\footnote{DICII Department, University of Rome ``Tor Vergata'', Via Politecnico 1, 00133 Rome, Italy. Email:\texttt{tomassetti@ing.uniroma2.it}}}
\begin{document}
\begin{sloppypar}

\maketitle

\begin{abstract}
Within the linearized three-dimensional theory of polymer gels, we consider a sequence of problems formulated on a family of cylindrical domains whose height tends to zero. We assume that the fluid pressure is controlled at the top and bottom faces of the cylinder, and we consider two different scaling regimes for the diffusivity tensor. Through asymptotic-analysis techniques we obtain two plate models where the transverse displacement is governed by a plate equation with an extra contribution from the fluid pressure.
In the limit obtained within the first scaling regime the fluid pressure is affine across the thickness and hence it is determined by its instantaneous trace on the top and bottom faces. In the second model, instead, the value of the fluid pressure 
is governed by a three-dimensional diffusion equation.
\end{abstract}

\medskip


\section{Introduction}
A polymer gel is a network of cross-linked polymeric chains permeated by a diffusing fluid. Temporal changes of fluid content bring about local swelling or de-swelling of the gel \cite{Doi2009JPSJa}, which in turn may determine substantial changes of shape. Two key properties affecting the amount of swelling are the compliance of the polymeric network and the affinity between the polymer and the fluid. Since these properties can be finely tuned during gel synthesis, understanding how they affect shape changes is crucial to controlling and harnessing shape changes.

Swelling--induced shape changes can be particularly dramatic in thin bodies, such as plates or rods \cite{GorieB2005PRL}. For example, the experiments reported in \cite{MoraB2006TEPJE} feature elastic instabilities during swelling of composite polymeric plates. The paper \cite{Holmes2011SoftMatter} investigates experimentally--observed time dependent bending of a thin plate of polymer gel during the transient that follows the deposition of a solvent droplet on one of its faces.

Mechanical theories of polymer gels and, more generally, of strained solids through which diffusion of a fluid takes place \cite{BaekSrinivasa2004,HongZZS2008JMPS,DudaSF2010JMPS,ChestA2010JMPS} have been elaborated drawing from the work of Gibbs concerning the role of chemical potential to describe the interaction between a solid and a fluid phase \cite{Gibbs1878} and from the work of Biot  on poroelastic media \cite{Biot1941Joap,Biot1972IUMJ}. In these theories, a polymeric gel is modeled as a single continuum where diffusion of a chemical species is driven by the gradient of its chemical potential. Constituive equations are obtained from a dissipation principle that takes explicitly into account the energetic flux associated to the motion of the fluid relative to the solid. A parallel line of research \cite{Bowen1980,ShiRajagopalWineman1981,CaldeCLZ2010JoCaTN} models the diffusing fluid and the strained solid as two superposed and interacting continua, whose governing equations result from the application of Truesdell's theory of mixtures \cite{Truesdell1962JCP}. It would be desirable to have a comparison of the two approaches, showing to what extent they are equivalent. A comparison in the setting of poroelasticity can be found in \cite{CallaA2011Tipm}.

In this paper we restrict our considerations to physical situations where the evolution of the polymer gel involves small departures from an equilibrium state. Following \cite{LucanN2012IJSS} we perform a formal linearization procedure of the mechanical equations governing motion and diffusion in a polymer gel proposed in \cite{DudaSF2010JMPS,ChestA2010JMPS}. We then arrive at the following constitutive equations:
\begin{equation}\label{eq:99}
\begin{aligned}
&\bm S=\mathbb C\bm E\bm u-p\bm I,\qquad \bm h=-\nu\bm M\nabla p,
\end{aligned}
\end{equation}
which govern the stress increment $\bm S$ with respect to the reference state and the flux of diffusant $\bm h$. Here, $\mathbb C$ is a fourth--order elasticity tensor acting on the linear strain $\bm E\bm u$;  $\nu$ is the molar volume of the fluid, $\bm M$ is the mobility tensor; the pressure $p$ entering Fick's law $\eqref{eq:99}_2$ is proportional to the fluid's chemical potential and plays the role of Lagrange multiplier associated to the incompressibility constraint
\begin{equation}\label{eq:100}
\operatorname{tr}(\bm E\bm u)=\nu\gamma  
\end{equation}
between the the increment $\gamma$ of the fluid concentration and the local volume change, as measured by the divergence $\operatorname{div}\bm u=\operatorname{tr}\bm E\bm u$ of the displacement field $\bm u$. The constitutive equations \eqref{eq:99} are combined with the balance equations:
\begin{equation}\label{eq:101}
\begin{aligned}
  &\operatorname{div}{\bm S}+\bm f=\bm 0,\qquad \dot \gamma+\operatorname{div}\bm h=0
\end{aligned}
\end{equation}
to obtain an evolution system governing displacement $\bm u$ and pressure $p$. In problems concerning deformations coupled with diffusion, inertial forces are typically neglected. Then, on identifying $\bm f$ with a prescribed body force of dead--load type, and on ruling out the unknowns $\bm S$ and $\bm h$ through their constitutive equations, we arrive at the following system in the uknowns $\bm u$ and $p$:
\begin{subequations}\label{eq:40}
 \begin{align}
  &\mathop{\rm div}\mathbb C\bm E\bm  u-\nabla p=-\bm f,
\\
  &\mathop{\rm div}\dot\bfu-\mathop{\rm div}\bm K\nabla p=0,
\end{align}
\end{subequations}
where $\bm K=\nu^2\bm M$ is the diffusivity tensor.

It is worth noticing that the constitutive equation $\eqref{eq:99}_1$ can be formally obtained as a specialization of the equation governing the stress of a saturated linear poroelastic medium \cite[Eq. 55]{DetournayCheng1993} in the case when the effective stress coefficient is equal to one. Moreover, the incompressibility constraint \eqref{eq:100} can be interpreted as the limit of the response equation for the pore fluid \cite[Eq. 59]{DetournayCheng1993} in the limit when the Biot modulus tends to infinity. In fact, to a certain extent, the equations of polymer gels can be considered as a specialization of the theory of saturated poroelastic media: Biot himself remarks in \cite{Biot1972IUMJ} that his ``pressure function'' applies to \emph{``[...] a broader context where it plays the role of ``chemical potential''}, whose applications are \emph{``[...] not restricted to the presence of actual pores. The fluid may be in solution with the solid, or may be adsorbed. Such phenomena are usually associated with the concept of capillarity or osmotic pressures.''} It is also worth noticing that the linear system we consider can be also considered as a special case of the equations of linear thermoelasticity \cite[Eq. 7.24]{Carlson1973LinearThermoelasticity} in the case when heat capacity vanishes and the stress--temperature tensor is equal to the identity. Such analogy has been explicitly reckoned, for instance, in \cite{Biot1964JoAM} and in \cite{Holmes2011SoftMatter}.

These observations put into perspective a number of theories concerning poroelastic and thermoelatic plates as relevant in the context we are considering. An earlier combination of poroelasticity with a structural theory was addressed by Biot, who considers in \cite{Biot1964JoAM} the effect of fluid flow within a poroelastic slab under compression, and who shows that creep buckling takes place when the compression is above a ``lower'' critical value, and that if the load is increased then the rate of lateral deflection increases until it becomes infinite at an ``upper'' critical value. In \cite{Taber1992Joam} Taber proposes a theory describing the coupling between diffusion and bending in a poroelastic plate, based on the Kirchhoff--Love kinematics, on the plane-strain assumption, and on the assumption that diffusion of fluid takes place only in the direction orthogonal to the mid-plane of the plate. 

A theory that specifically addresses thin polymer gels coupled with fluid diffusion has more recently been proposed in \cite{LucanN2012IJSS} concerning a planar beam whereby motion and diffusion are restricted to two dimensions. In that paper, the constitutive equations governing the relation between bending moment and curvature are obtained from the standard kinematic assumptions that the cross section of the beam remain straight, and that the stress along the transverse direction vanishes. Within these assumptions, two different models have been proposed: a two--dimensional diffusion model which allows the fluid to diffuse both in the axial and in the transverse direction of the beam; a one--dimensional diffusion model which allows fluid diffusion only in the transverse direction, as in the above--mentioned paper by Taber. Comparison between the two models shows that the one dimensional model is still adequate to capture the response of the beam to pressure variations. More recently, the equations of a poroelastic plate under large strain have been derived using the virtual-powers approach combined with an enriched kinematics \cite{lucantonio2016large}.  

All the above--mentioned derivations make use of \emph{ad hoc} hypotheses to obtain a structural model. There is however a number of papers dealing with poroelasticity and thermoelasticity  which make use of rigorous asymptotic analysis to obtain structural (plate or beam) models, in the spirit of \cite{CiarlD1979JM}: one considers a \emph{generating family} of initial-boundary value depending on a vanishing parameter $\eps$, designed in a way to capture the thinness of the body under scrutiny (be it a plate or a beam), and studies the asymptotic behavior of the solutions of these problems as $\eps$ tends to null. One may naively expect that such a procedure, being rigorous, produces a definite answer concerning what structural model betters capture thinness. However, this is not the case: as pointed out in \cite{ParonP2015JoE}, the outcome of the asymptotic method depends crucially on the \emph{choice} of the generating family, and different choices lead to different structural models. For example, there is more than one choice that can lead to the Reissner--Mindlin plate model \cite{ParonPT2007AA,ParonPT2006CRMASP,NeffHJ2008Reissner,RieyT2008JCA}, to models used for plate buckling \cite{ParonT2013MMS}, plates with residual stress \cite{Paron2006M,Paron2006MMSa} or to the Timoshenko beam theory \cite{RieyT2009CAA,FalacPP2015AaA}.

The generating families are usually obtained by writing the relevant equations (in our case, \eqref{eq:99}--\eqref{eq:101}) on a collection of cylinders
\begin{equation}
\Omega_\eps=\omega\times\left(-\eps \frac h 2,+\eps \frac h 2\right),\qquad \omega\subset\mathbb R^2,
\end{equation}
which shrink to the planar domain $\omega\times\{0\}$ as $\varepsilon$ tends to null. The coefficients in these equations, as well as the boundary conditions, can possibly depend on the parameter $\eps$. Most of these theories are essentially a variant of the Kirchhoff--Love plate equation or Euler--Bernoulli beam equation with an extra contribution to the bending moment coming from pressure inhomogeneities across the thickness. This contribution evolves in time as determined by the flow of fluid or thermal conduction within the plate and is itself affected by bending.

In the simplest setting of three-dimensional linear elasticity, the most popular approach to the derivation of  plate equations by asymptotic analysis is to keep the elasticity tensor $\mathbb C$ independent of $\varepsilon$, and replace the displacement field $\bm u_\varepsilon$ with the scaled displacement $\bm u^\varepsilon$ whose components are defined as (\emph{cf.} Eq. 74 in Sec. \ref{sec:plate-equations}):
\begin{equation}\label{eq:42}
  u^\eps_\alpha=\frac 1 \eps (\bm u_\eps)_\alpha,\qquad u^\eps_3=(\bm u_{\varepsilon})_3. 
\end{equation}
It appears natural to replicate this pattern when mechanical effects are coupled with diffusion, and a simple calculation (see the chain of argument that leads from Eq. \ref{eq:75} to Eq. \ref{eq:9}) shows that if one wants the pressure $p$ to appear in the limit equation, then the original unknown $p_\eps$ should be rescaled by introducing the new unknown
\begin{equation}
  p^\eps=\frac{p_\eps}{\eps}.
\end{equation}
These scaling are essentially common to all asymptotic derivations we were able to find in the literature. However, there are still many degrees of freedom in the choice of the generating family of problems. This explains why so many different limit theories can be obtained. For example, the derivation of theories of thermoelastic plates presented in \cite{Lebel1992CRASPSIM}, where temperature plays the role of pressure, makes use of Robin--type conditions on the top and bottom faces of the cylinder and considers three different scalings, which result in different limit problems. The paper \cite{Marciniak-Czochra2015}, which is concerned with poroelastic plates, considers the regime where the diffusivity tensor $\bm K$ scales as $\varepsilon^2$. In this case, two systems of partial differential equations are obtained: the first is an elliptic system formulated on the 2D domain $\omega$, which governs the in--plane components of displacement and the average pressure across the thickness; the second is a system composed of the classical plate equation formulated on $\omega$, coupled with a particular diffusion equation formulated on the domain $\Omega_1$ whose peculiarity is that diffusion can take place only in the transverse direction. The result is essentially the same as the 1D stress-diffusion model discussed in \cite{LucanN2012IJSS}.

In this paper we consider material symmetry of monoclinic type with respect to the plane containing $\omega$, and we build two different generating families. For the first family, the diffusivity tensor $\bm K$ is independent of $\varepsilon$, so that the only datum that depends on $\eps$ is the thickness of the domain $\Omega_\eps$. We show in Section \ref{sec:plate-equations} that this family produces a two dimensional model where the value of the pressure across the thickness is the affine interpolation of the values of the pressure at the top and bottom face, as in theories describing thermal bending. For the second family we consider, the planar components of the diffusivity tensor are constant, while the transverse component scales as $\varepsilon^2$. As we show in Section \ref{sec:diff-plate-model}, this family generates the diffusive model proposed in \cite{LucanN2012IJSS}.

\section{A nonlinear system of evolution equations}\label{sec:evolution-problem}
In this section, we put together the system of partial differential equations that arises in mechanical theories that describe diffusion of a solvent through a finitely--strained solid, and we linearize the resulting equations to arrive at an initial--boundary value problem that is the object of our asymptotic analysis. For the sake of coinciseness, we limit ourselves to presenting the key ingredients and we refer the reader eager for further details to the many presentation available in the literature, which we have listed in the Introduction. Moreover, we leave for the next section the specification of the relevant boundary and initial conditions.

As a start we recall, that for $\Omega$ the region occupied by the body in its reference configuration, and for $I=[0,T]$ the time interval of interest, the unknown fields for the theory in question are: the \emph{deformation} $f$, the \emph{Piola stress} $\bm S$, the \emph{referential concentration of solvent} $c$ (a posisitive--valued scalar field), the \emph{referential solvent flux} $\bm h$, and the \emph{solvent chemical potential} $\mu$.

For $\bm f$ the dead body--force field, the aforementioned fields are required to satisfy the force-- and the mass--balance:
\begin{subequations}\label{eq:39}
\begin{align}
&\textrm{div}\bm S+\bm f=\bm 0,\label{eq:39a}\\
&\dot c+\textrm{div}\bm h=0,\label{eq:39b}
\end{align}
\end{subequations}
as well as a set of constitutive equations consistent with the dissipation inequality
\begin{align}\label{eq:20}
\dot\psi\le\bm S\cdot\dot{\bm F}+\mu\dot c-\bm h\cdot\nabla\mu,
\end{align}
where 
\begin{equation}\label{eq:38}
  \bm F=\nabla f
\end{equation}
is the deformation gradient.

For applications to gels, it is usually assumed that both the polymer and the solvent be incompressible. This assumption does not imply that the admissible deformations of the gel as a whole are isochoric: volume changes, as measured by the determinant  $\operatorname{det}\bm F$ of the deformation gradient, may take place due to changes of the relative volume proportions of polymer and solvent. However, these changes must comply with the \emph{incompressibility constraint}:
\begin{equation}\label{eq:119}
  \textrm{det}\bm F={1+}\nu(c-\mathring c),
\end{equation}
where $\mathring c$ is the number of solvent molecules per unit referential volume and $\nu$ is  molecular volume. For $\bm F^\star={\rm det }\bm F\ \bm F^{-T}$ the cofactor of $\bm F$, we have from \eqref{eq:119} that the rate of change of concentration is related to the rate of deformation gradient by $\displaystyle{\dot c=\nu^{-1}{\bm F^\star}\cdot\dot{\bm F}}$. As a result, \eqref{eq:20} yields
\begin{align}\label{eq:35}
  \dot\psi\le\Big(\bm S+\frac\mu\nu\bm F^\star\Big)\cdot\dot{\bm F}-\bm h\cdot\nabla\mu.
\end{align}
By comparing \eqref{eq:20} with \eqref{eq:35} we see that an important consequence of the incompressibility constraint \eqref{eq:119} is that the Piola stress $\bm S$ expend power in concomitance with the chemical potential $\mu$, and hence only a combination of these fields,  namely $\bm S +\mu/\nu \bm F^\star$, can be constitutively prescribed.
Thus, in view of \eqref{eq:35}, we adopt the following constitutive equations:\color{black}
\begin{subequations} \label{eq:23}
\begin{align}
  &\bm S={\psi'}(\bm F)-\frac \mu \nu\bm F^\star\label{eq:25ba}\\
  &\bm h=-\bm M(c)\nabla \mu,\label{eq:25c}
\end{align}
\end{subequations}
where $\bm M(c)$ is a concentration--dependent, positive definite  mobility tensor.
Viewing \eqref{eq:25ba} as a constitutive equation for $\bm S$ leaves the chemical potential not  specified constitutively. In fact, within this theory, the chemical potential cannot be determined only by solving the complete problem. The mechanical interpretation of $\mu$ comes through the decomposition 
\begin{equation}\label{eq:22}
  \mu=\mu^0+\nu\sigma,
\end{equation}
where $\mu^0$ is the \emph{chemical potential of the pure solvent}, and $\sigma$ is the \emph{solvent pressure}.

When put together, the balance equations \eqref{eq:39}, the compatibility condition \eqref{eq:38}, the incompressibility constraint \eqref{eq:119}, and the constitutive equations \eqref{eq:23} form a system of partial differential equations in the unknowns $(f,\bm F,\bm S,c,\mu,\bm h)$. 

\section{Linearization of the evolution equations}
Our next task is to perform a formal linearization of this system about a suitable reference state. As a first step, we take as reference state an equilibrium solution of the aforementioned system, where the body is undeformed, so that $\mathring f(x)=x$ and $\mathring{\bm F}=\bm I$, and both stress and solvent flux vanish. We stick to the convention of marking by a superimposed ring all fields that pertain to the reference state. Accordingly,  by \eqref{eq:23} we write
\begin{subequations}\label{eq:25}
\begin{align}
&\mathring{\bm S}=\psi'(\bm I)-\frac {\mathring\mu}\nu\bm I=\bm 0,\label{eq:25a}\\
&\mathring{\bm h}=\bm M(\mathring c)\nabla\mathring\mu=\bm 0.\label{eq:25b}
\end{align}
\end{subequations}
In particular, in the reference state the solvent pressure is
\begin{equation}\label{eq:37}
   \mathring\sigma=\frac{\mathring\mu-\mu^0}\nu.
\end{equation}
As our second step we consider small departures from the reference state and we write
\begin{equation}\label{eq:105}
f(x)=x+\bm u(x),\qquad \bm F=\bm I+\bm H,\qquad c=\mathring c+\gamma,\qquad \mu=\mathring\mu+\nu p,
\end{equation}
where $\bm u$ is the \emph{displacement},  $\bm H=\nabla\bm u$ is the \emph{displacement gradient},  $\gamma$ is the \emph{concentration increment}, and (\emph{cf.} \eqref{eq:22} and \eqref{eq:37})
\begin{equation}
  p=\sigma-\mathring\sigma,
\end{equation}
is the \emph{solvent pressure increment}.

In the following we  assume that the loads and the boundary conditions are such that
$|\bm H|$ is small and the increment of pressure solvent $p$ is also of the same order.
From the incompressibility constraint \eqref{eq:119}, it follows that also the concentration increment $\gamma$ is of order $|\bm H|$.

The following linearization formula:
\begin{equation}
  \label{eq:34}
  \bm F^{\star}\cong\bm I+(\bm I\cdot\bm H)\bm I-\bm H^T,
\end{equation}
when applied to \eqref{eq:25ba} allows us to write:
$$
  \bm S\cong\psi'(\bm I)+\psi''(\bm I)[\bm H]-\frac{\mathring \mu}\nu(\bm I+(\bm I\cdot\bm H)\bm I-\bm H^T)-p\bm I,
$$
where the symbol $\cong$ indicates that the equality hods up to infinitesimals of first order. 
By making use of the representation formula
\begin{equation}
  \label{eq:36}
  \psi''(\bm I)[\bm H]=
  \textrm{sym}(\psi''(\bm I)[\bm E])+\bm H\psi'(\bm I)-{\rm sym}(\bm E\mathbf\psi'(\bm I)),
 \qquad \bm E={\rm sym}\bm H,
\end{equation}
which follows from frame indifference, and recalling \eqref{eq:25a}
we find that
\begin{equation}
  \label{eq:35}
  \begin{aligned}
    \bm S&\cong\psi'(\bm I)-\frac{\mathring \mu}\nu\bm I+\textrm{sym}(\psi''(\bm I)[\bm E])+\bm H\psi'(\bm I)-{\rm sym}(\bm E\mathbf\psi'(\bm I))\\
    &\qquad\qquad-\frac{\mathring \mu}\nu((\bm I\cdot\bm H)\bm I-\bm H^T)-p\bm I,\\
    &=\textrm{sym}(\psi''(\bm I)[\bm E])+\frac{\mathring\mu}\nu\bm H-{\rm sym}(\bm E\mathbf\psi'(\bm I))-\frac{\mathring\mu}\nu((\bm I\cdot\bm H)\bm I-\bm H^T)-p\bm I,\\
    &=\textrm{sym}(\psi''(\bm I)[\bm E])+\frac{\mathring\mu}\nu(\bm E-(\bm I\cdot\bm E)\bm I)-p\bm I.
    \end{aligned}
  \end{equation}
  On introducing the fourth--order tensor $\mathbb C$ defined by
  \begin{equation}
  \mathbb C[\bm E]=\textrm{sym}(\psi''(\bm I)[\bm E])+\frac{\mathring \mu}\nu (\bm E-(\bm I\cdot\bm E)\bm I).
\end{equation}
we can rewrite \eqref{eq:35} as
\begin{equation}
  \label{eq:11135}
\bm S\cong\mathbb C[\bm E]-p\bm I.
\end{equation}
Next,  we notice that \eqref{eq:25b}, together with the positive definiteness of the mobility tensor, implies\color{black}
\begin{equation}
  \nabla\mathring\mu=0,
\end{equation}
hence, when we linearize\eqref{eq:35} we obtain,
\begin{equation}
  \label{eq:135}
 \begin{aligned}
   \bm h&\cong-\bm M(\mathring c)\nabla \mathring \mu-\gamma\bm M'(\mathring c)\nabla \mathring\mu-\nu\bm M(\mathring c)\nabla p\\
   \color{blue}&=-\nu \bm M(\mathring c)\nabla p.
   \end{aligned}
\end{equation}
By putting together the balance equations \eqref{eq:39}, the linearized constitutive equations \eqref{eq:11135} and \eqref{eq:135}, and the linearization of the incompressibility constraint \eqref{eq:119}, we obtain the following \emph{incremental system}:
\begin{equation}\label{eq:107}
  \begin{aligned}
    &\operatorname{div}\bm S+\bm b=\bm 0\\
    &\bm S\cong\mathbb C[\bm E\bm u]-p\bm I\\
   &\dot \gamma\cong\nu\operatorname{div}(\bm M(\mathring c)\nabla p)\\
  &\bm I\cdot \bm E\bm u\cong1+\nu\gamma\\
  &\bm  E\bm u={\rm sym}\nabla\bm u
  \end{aligned}\qquad\text{in }\Omega,
\end{equation}
where we wrote explicitly the dependence on the displacement $\bm u$ of the strain $\bm E\bm u=\textrm{sym}\nabla\bfu$.
Our final step consists in the elimination of the unknowns $\bm S$ and $\gamma$ from \eqref{eq:107}, a step that leads to the following system in the unknowns $\bm u$ and $p$: 
\begin{subequations}
\begin{align}\label{eq:43}
 &\textrm{div}\,\mathbb C[\bm E\bm u]-\nabla p=-\bm f,\\
  &\textrm{div}\,\dot{\bm u}=\nu^2\textrm{div}\bm M(\mathring c)\nabla p.
\end{align}
\end{subequations}
From now on we set $\bm K=\nu^2\bm M(\mathring c)$.

\section{The existence an uniqueness of the solution of the three-dimensional problem}
The linearization procedure carried out in the previous section leads to the following system of equations\color{black}
\begin{subequations}\label{eq:1}
 \begin{align}
   &\mathop{\rm div}\mathbb C\bm E\bm  u-\nabla p=-\bm f,
\\
  &\mathop{\rm div}\dot\bfu-\mathop{\rm div}\bm K\nabla p=0,
\end{align}
\end{subequations}
formulated on a space--time region $\Omega\times(0,T)$, where $\Omega$ is a domain of $\mathbb R^3$. 

For the application to plates we have in mind, the following combination of boundary conditions is relevant:
\begin{subequations}\label{eq:2}
\begin{align}
\bfu(t)&=\bm 0\textrm{ on }\Gamma_{u,D},&(\mathbb C\bm E\bm u(t))\bm n-p(t)\bm n&={\bm 0}\textrm{ on }\Gamma_{u,N}, \label{eq:2a}\\
p(t)&=p_{\rm a}(t)\textrm{ on }\Gamma_{p,D},&-\bm K\nabla p(t)\cdot\bm n&=\bm 0\textrm{ on }\Gamma_{p,N}
,\label{eq:2b}
\end{align}
\end{subequations}
where $\partial\Omega=\Gamma_{u,D}\cup\Gamma_{u,N}$ and $\partial\Omega=\Gamma_{p,D}\cup\Gamma_{p,N}$ are partitions of $\partial\Omega$. Here  $p_{\rm a}$ is a  pressure increment applied $\Gamma_{p,D}$. The formulation of an initial--boundary value problems is arrived at by imposing an \color{black}initial condition for the pressure:\color{black}
\begin{align}\label{eq:3}
  p(0)=\color{black}p_0\color{black},
\end{align}
with $p_0$ compatible with \eqref{eq:2b}. 

We dedicate the rest of this section to proving the existence and the uniqueness of solutions for Problem \eqref{eq:1}--\eqref{eq:3} in the case when the applied pressure increment vanishes: 
\begin{equation}\label{eq:48}
p_{\rm a}=0.
\end{equation}
We leave to a remark at the end of this section the illustration of how to treat the more general case where the applied pressure $p_{\rm a}$ does not vanish.

\medskip
\noindent {\bf Notation.} We denote by $\mathbb R^{3\times 3}_{\rm Sym}$ the vector space of symmetric 3$\times$3 matrices with real entries.  We introduce the function spaces
\[
H^1_{u,D}(\Omega;\mathbb R^3)=\{\bm v\in H^1(\Omega;\mathbb R^3):\bm v=0\text{ on }\Gamma_{u,D}\}
\]
and
\[
H^1_{p,D}(\Omega)=\{q\in H^1(\Omega):q=0\text{ on }\Gamma_{p,D}\},
\]
respectively, for the displacement $\bm u$ and the pressure $p$.  \color{black}For $B$ a Banach space, we denote by $L^p(0,T;B)$ the standard Lebesgue spaces of Bochner-integrable functions defined on the time interval $(0,T)$ and taking values in $B$. Likewise, we denote by  $H^1(0,T;B)$ the corresponding Sobolev space of functions whose derivative with respect to time is in $L^2(0,T;B)$. For typographical convenience, when denoting the Lebesgue or Sobolev norm of functions defined on $\Omega$ and taking values in $\mathbb R^N$, we omit the specification of the domain and the codomain. For example, we write $\|\cdot\|_{L^q}$ and $\|\cdot\|_{L^p(0,T;L^q)}$ in place of  $\|\cdot\|_{L^q(\Omega;\mathbb R^n)}$ and $\|\cdot\|_{L^p(0,T;L^q(\Omega;\mathbb R^n))}$, respectively.\color{black}

\medskip

\noindent {\bf Assumptions.} We shall make the following assumptions concerning the data:
\begin{subequations}\label{eq:31}
\begin{align}
&\color{black}\mathcal H^2(\Gamma_{u,D})>0,\mathcal H^2(\Gamma_{p,D})>0,\color{black}\\
&\color{black}p_0\color{black}\in H^1_{p,D}(\Omega),\\  
&\bm f\in H^1(0,T;L^2(\Omega;\mathbb R^3)),\label{eq:31c}\\
&\mathbb C_{ijkl}\in L^\infty(\Omega),\\
&\mathbb C_{ijkl}=\mathbb C_{klij}=\mathbb C_{jikl},\\
&\text{$\exists c_{\mathbb C}>0:\mathbb C\bm A\cdot\bm A\ge c_{\mathbb C}|\bm A|^2\quad\forall \bm A\in\mathbb R^{3\times 3}_{\rm Sym}$}\quad \text{a.e. in }\Omega,\label{eq:31d}\\
&\mathbf K\in L^\infty(\Omega;\mathbb R^{3\times 3}_{\rm Sym}),\\
&\text{$\exists c_{\bm K}>0:\mathbf K\bm a\cdot\bm a\ge c_{\bm K}|\bm a|^2\quad\forall \bm a\in\mathbb R^3$}\quad \text{a.e. in }\Omega.\label{eq:31e}
\end{align}
\end{subequations}

\begin{theorem}\label{thm:main}
Under assumptions \eqref{eq:31}, the initial-boundary-value problem \color{black}\eqref{eq:1}-\eqref{eq:3} has a unique solution in the following weak sense:
\begin{equation}\label{eq:27}
\left\{ 
\begin{aligned}
  &\bfu\in H^1(0,T;H^1_{u,D}(\Omega;\mathbb R^3)),\\
&p\in L^\infty(0,T;H^1_{p,D}(\Omega))\cap H^1(0,T;L^2(\Omega))\quad\text{ with } \quad p(0)=\color{black}p_0,\\
&\int_{\Omega}\big(\mathbb C\bm E\bfu(t)\cdot\bm E\bm v- p(t)\mathop{\rm div}\bm v)\mathop{{\rm d}x}=\int_{\Omega}{\bm f}(t)\cdot\bm v\mathop{{\rm d}x}\\
&\qquad\forall\bm v\in H^1_{u,D}(\Omega;\mathbb R^3),\color{black}\forall t\in[0,T),\\
&\int_\Omega \left({\rm div}\dot\bfu(t)\,q+\bm K\nabla p(t)\cdot\nabla q\right)\mathop{{\rm d}x}=0\quad \forall q\in H^1_{p,D}(\Omega),\text{ for a.e. }t\in(0,T).\\
\end{aligned}\right. 
\end{equation}
\end{theorem}

We perform the proof of existence in several steps.\medskip

\noindent \subsection{Time discretization.}\medskip

\noindent Our first step is the construction of a sequence of approximate solutions to \eqref{eq:1}-\eqref{eq:2}. Our approximation scheme is based on a time discretization with time step $\tau=T/n$, with $n$ an integer that we shall eventually let to $+\infty$.

With a view towards setting up our iterative scheme, we begin by approximating the bulk datum $\bm f$, with the sequence 
\begin{align}\label{eq:32}
  \bm f\tk=\frac 1 \tau \int_{(k-1)\tau}^{k\tau}\bm f(t){\rm d}t\in L^2(\Omega;\mathbb R^3)\quad\textrm{ for $k=1\dots n$},
\end{align}
and we set 
$$
\bm f^0_\tau=\bm f(0),
$$
where $\bm f(0)$, understood in the sense of traces, is well defined thanks to the time regularity of the forcing term $\bm f$ stipulated in Assumption $\eqref{eq:31c}$.

Next, we define the initial displacement $\bm u_0$ as the unique solution of the variational problem:
\begin{equation}\label{eq:28}
\left\{
\begin{aligned}
&\bm u_0\in H^1_{u,D}(\Omega;\mathbb R^3),\\
&\int_\Omega \left(\mathbb C\bfE\bfu_0\cdot\bfE\bm v-p_0\mathop{\rm div}\bm v\right){\rm d}x=\int_\Omega \bm f^0_\tau\cdot\bm v\mathop{{\rm d}x}\quad\forall \bm v\in H^1_{u,D}(\Omega;\mathbb R^3).\\
\end{aligned}
\right.
\end{equation}
The existence and the uniqueness of the solution of \eqref{eq:28} follows immediately from the Lax--Milgram lemma. Now, for
\begin{equation}\label{eq:6}
\bm u^0_\tau=\color{black}\bm u_0,\qquad p_\tau^0=p_0
\end{equation}
we formulate recursively, starting from $k=1$ up to $k=n$, a sequence of problems consisting in the system of partial differential equations:
\begin{subequations}\label{eq:4}
\begin{align}
  &\mathop{\rm div}\mathbb C\bfE\bfu\tk-\nabla p\tk=-\bm f\tk,\label{eq:4a}\\
  &\mathop{\rm div}\left(\frac{\bfu\tk-\bfu\tkk}\tau\right)-\mathop{\rm div}\bm K\nabla p\tk=0,
\end{align}
\end{subequations}
in the unknows $\bm u\tk$ and $p\tk$, supplemented by the boundary conditions:
\begin{subequations}\label{eq:5}
\begin{align}
&\bfu\tk=0\textrm{ on }\Gamma_{u,D},\qquad \mathbb C(\bm E\bm u\tk)\bm n-p\tk\bm n=0\textrm{ on }\Gamma_{u,N}, \label{eq:5a}\\
&p\tk=0 \textrm{ on }\Gamma_{p,D},\qquad \bm K\nabla p\tk\cdot \bm n=0\textrm{ on }\Gamma_{p,N}.
\end{align}
\end{subequations}
Our next step is to establish the existence of a weak solution to \eqref{eq:4}-\eqref{eq:5} for all $k$.

\begin{proposition}
For $k\in \{1,\dots,n\}$, let $\bm u_\tau^{k-1}\in H^1_{u,D}(\Omega;\mathbb R^3)$ be given. Then the boundary-value problem \eqref{eq:4}-\eqref{eq:5} has a unique solution in the following sense:
\begin{equation}\label{eq:50}
\left\{ 
\begin{aligned}
&\bm u_\tau^k\in H^1_{u,D}(\Omega;\mathbb R^3),\\
&p_\tau^k\in H^1_{p,D}(\Omega),\\
&\int_\Omega \left(\mathbb C\bfE\bfu\tk\cdot\bfE\bm v-p\tk\mathop{\rm div}\bm v\right){\rm d}x=\int_\Omega \bm f\tk\cdot\bm v\mathop{{\rm d}x}\quad\forall \bm v\in H^1_{u,D}(\Omega;\mathbb R^3),\\
&\int_\Omega \left(\mathop{\rm div}\bfu\tk\, q+\tau\bm K\nabla p_\tau^k\cdot \nabla q\right){\rm d}x=\int_\Omega \mathop{\rm div}\bm u_\tau^{k-1}\,q\mathop{{\rm d}x}\quad \forall q\in H^1_{p,D}(\Omega).
\end{aligned}\right.  
\end{equation} 
\end{proposition}
\begin{proof}
We introduce the Hilbert space ${H}=H^1_{u,D}(\Omega;\mathbb R^3)\times H^1_{p,D}(\Omega)$ equipped with the scalar product
\begin{align}
  \big((\bm u,p),(\bm v,q)\big)_{{H}}=(\bm u,\bm v)_{H^1}+(p,q)_{H^1}.
\end{align}
We define the bilinear form $a_\tau:{H}\times{H}\to\mathbb R$ and the linear functional $\ell_\tau^k:{H}\to\mathbb R$ defined by  
\begin{align}
  a_\tau\big((\bm u,p),(\bm v,q)\big)=\int_{\Omega} \big(q \mathop{\rm div}\bm u+\tau\bm K\nabla p\cdot\nabla q+\mathbb C\bm E\bm u\cdot\bm E\bm v-p\mathop{\rm div}\bm v\big)\mathop{{\rm d}x},
\end{align}
and 
\begin{equation}
  \ell_\tau^k(\bm v,q)=\int_\Omega \big(\bm f_\tau^k\cdot\bm v+\mathop{\rm div}\bm u_\tau^{k-1}\,q\big)\mathop{{\rm d}x}.
\end{equation}
Problem \eqref{eq:50} can then be rewritten as follows:
\begin{equation}\label{eq:51}
  a_\tau\big((\bm u_\tau^k,p_\tau^k),(\bm v,q)\big)=\ell_\tau^k(\bm v,q)\quad\forall (\bm v,q)\in{H}.
\end{equation}
Since 
\begin{align}
   a_\tau((\bm u,p),(\bm v,q))\le C\|(\bm u,p)\|_{{H}}\|(\bm v,q)\|_{{H}}\quad\forall (\bm u,p),(\bm v,q)\in{H},
\end{align}
the bilinear form $a_\tau$ is continuous. Moreover, by Korn's and Poincar\'e's inequalities, we have
\begin{equation}
\begin{aligned}
a_\tau((\bm u,p),(\bm u,p))&=\int_\nu \big(\tau\bm K\nabla p\cdot\nabla p+\mathbb C
\bm E\bm u\cdot\bm E\bm u\big)\mathop{{\rm d}x}\\
&\ge \tau C_{\bm K} \|\nabla p\|_{L^2}^2+C_{\mathbb C}\|\bm E\bm u\|^2_{L^2}\\
&\ge  \tau C \|(\bm u,p)\|_{{H}}
\end{aligned}\quad\forall (\bm u,p)\in{H},
\end{equation}
hence the bilinear form $a_\tau(\cdot,\cdot)$ is coercive for each $\tau$. The existence of a unique solution to \eqref{eq:51} follows from the Lax-Milgram Lemma.
\end{proof}

\noindent\subsection{A priori estimates for the time-discrete problem.}\medskip

\noindent We shall make use of the following discrete version of Gronwall's inequality (see \cite[Lemma 1.4.2]{QuartV2008Numerical}):

\begin{lemma}\label{lem:1}
Let $\{a_k\}$ and $\{b_k\}$ be non-negative sequences and let $c_0>0$. If the sequence $\{y_k\}$ satisfies
\begin{equation}
 y_k\le c_0+\sum_{j=1}^{k-1}(a_j+b_jy_j),
\end{equation}
then
\begin{equation}
 y_k\le \big(c_0+\sum_{j=1}^{k-1}a_j\big) e^{{\sum_{j=1}^{k-1}b_j}}.
\end{equation}
 \end{lemma}
Moreover, we shall make use of the following result.
\begin{lemma}\label{lem:11}For $k=1\dots n$, the interpolants $\bm f^k_\tau$ defined in \eqref{eq:32} satisfy the inequalities
\[
\tau\sum_{k=1}^n\left\|\frac{\bm f^k_\tau-\bm f^{k-1}_\tau}\tau\right\|^2_{L^2}\le C,\qquad \|\bm f^k_\tau\|_{L^2}\le C,
\]
with $C$ a constant that does not depend on $\tau$.
\begin{proof}
On exploiting the time-regularity of $\bm f$ stipulated in Assumption \eqref{eq:31c} we deduce the following chain of equalities and inequalities for $k\ge 1$:
\begin{equation}\label{eq:55}
\begin{aligned}
\nonumber  \|\bm f_\tau^{k+1}-\bm f\tk\|^2_{L^2}&=\Big\|\frac 1 \tau \int_{k\tau}^{(k+1)\tau}\int_{t-\tau}^t\dot{\bm f}(s){\rm d}s{\rm d}t\Big\|^2_{L^2}\le\frac 1 \tau \int_{k\tau}^{(k+1)\tau}\Big\|\int_{t-\tau}^t\dot{\bm f}(s){\rm d}s\Big\|^2_{L^2}{\rm d}t
\nonumber
\\
&=\tau\int_{k\tau}^{(k+1)\tau}\Big\|\frac 1 \tau \int_{t-\tau}^t\dot{\bm f}(s){\rm d}s\Big\|^2_{L^2}{\rm d}t\le \tau\int_{(k-1)\tau}^{(k+1)\tau}\big\|\dot{\bm f}(s)\big\|^2_{L^2}{\rm d}s.
\end{aligned}
\end{equation}
By the same token, we have
\begin{align}
\|\bm f_\tau^1-\bm f_\tau^0\|^2_{L^2}&=\left\|\frac 1\tau\int_0^\tau (\bm f(t)-\bm f(0)){\rm d}t\right\|^2_{L^2}=\left\|\frac 1\tau\int_0^\tau \int_0^t \dot{\bm f}(s){\rm d}s{\rm d}t\right\|^2_{L^2}\nonumber\\
&\le \frac 1 \tau \int_0^\tau\left\|\int_0^t \dot{\bm f}(s){\rm d}s\right\|^2_{L^2}{\rm d}t\nonumber\\
&\le \tau\int_0^\tau \big\|\dot{\bm f}(s)\big\|^2_{L^2}{\rm d}s.\nonumber
\end{align}
On combining the above inequalities we have the first inequality of the thesis. The second inequality is a consequence of the first inequality and of the following chain:
\begin{align}
\|\bm f_\tau^k\|_{L^2}&\le \|\bm f_\tau^0\|_{L^2}+\sum_{j=1}^k\|\bm f_\tau^j-\bm f_\tau^{j-1}\|_{L^2}\nonumber\\
&\le \|\bm f_\tau^0\|_{L^2}+\sum_{j=1}^n\|\bm f_\tau^j-\bm f_\tau^{j-1}\|_{L^2}\nonumber\\
&= \|\bm f_\tau^0\|_{L^2}+\sum_{j=1}^n\tau\left\|\frac{\bm f_\tau^j-\bm f_\tau^{j-1}}{\tau}\right\|_{L^2}\nonumber\\
&\le \|\bm f_\tau^0\|_{L^2}+T^{1/2}\left(\sum_{j=1}^n\tau\left\|\frac{\bm f_\tau^j-\bm f_\tau^{j-1}}{\tau}\right\|^2_{L^2}\right)^{1/2}.\nonumber
\end{align}
\end{proof}

\end{lemma}
\color{black}
\noindent \color{black}The estimates we are going to derive are best written in terms of the following interpolants:
\color{black}
\begin{equation}
\label{eq:30}
\left.
\begin{array}{l}
  \overline{\bm u}_\tau(t)=\bm u\tk,\\[0.5em]
  \overline{p}_\tau(t)=p\tk,\\
  \bm u_\tau(t)=\bm u_\tau^{k-1}+(t-(k-1)\tau)\frac{\bm u_\tau^k-\bm u_\tau^{k-1}}{\tau}\\[1ex]
   p_\tau(t)=p_\tau^{k-1}+(t-(k-1)\tau)\frac{p_\tau^k-p_\tau^{k-1}}{\tau}
\end{array}
\right\}
\begin{aligned}
&\text{ for $k=1,\dots,n$}\\
&\text{ and $t\in ((k-1)\tau,k\tau]$.}
\end{aligned}
\end{equation}

\begin{proposition}[Energetic estimates]\label{prop:2}
The estimates
\begin{subequations}\label{eq:15}
\begin{align}
  &\|\overline{\bm u}_\tau\|_{\color{black}L^\infty(0,T;H^1)}\le C,\\
  &\|\overline p_\tau\|_{\color{black}L^2(0,T;H^1)}\le C,
  \end{align}
\end{subequations}
hold uniformly with respect to $\tau$.
\end{proposition}
\begin{proof}
For $k\in\{1\dots n\}$ we \color{black}choose $\bm v=\bm u\tk-\bm u\tkk$ and $q=p\tk$ as test functions in the weak formulation \eqref{eq:50} of the discrete scheme, \color{black}and we add the resulting equations to obtain:
\begin{equation}
  \int_\Omega \left(\mathbb C\bfE\bfu\tk\cdot\bfE(\bm u\tk-\bm u\tkk)+\tau\bm K\nabla p_\tau^k\cdot \nabla p\tk\right){\rm d}x=\int_\Omega \bm f\tk\cdot(\bm u\tk-\bm u\tkk)\mathop{{\rm d}x}.
\end{equation}
The symmetry and the positivity of the elasticity tensor imply the inequality $2\mathbb C \bm A\cdot(\bm A-\bm B)\ge\mathbb C\bm A\cdot\bm A-\mathbb C\bm B\cdot\bm B$ for every pair of symmetric tensors $\bm A$ and $\bm B$. We therefore arrive at the following inequality:
\begin{equation}\label{eq:7}
\begin{aligned}
  &\frac 12 \int_\Omega \mathbb C\bm E\bm u\tk\cdot\bm E\bm u\tk\mathop{{\rm d}x}-\frac 12\int_\Omega \mathbb C\bm E\bm u\tkk\cdot\bm E\bm u\tkk\mathop{{\rm d}x}+\tau\int_\Omega \bm K\nabla p\tk\cdot\nabla p\tk\mathop{{\rm d}x}
\\
&\quad \le\int_\Omega \bm f\tk\cdot(\bfu\tk-\bfu\tkk)\mathop{{\rm d}x}.
\end{aligned}
\end{equation}
Let $1\le j\le n=T/\tau$. From \eqref{eq:7}, by performing summation of both sides with $k$ running from $1$ to $j$ we obtain
\begin{multline}\label{eq:8}
  \frac 12\int_\Omega\mathbb C\bm E\bm u_\tau^j\cdot\bm E\bm u_\tau^j\mathop{{\rm d}x}+\sum_{k=1}^j \tau\int_\Omega \bm K\nabla p\tk\cdot\nabla p\tk\mathop{{\rm d}x}\\
\le \frac 12\int_\Omega \mathbb C\bm E\bm u_\tau^0\cdot\bm E\bm u_\tau^0\mathop{{\rm d}x}+\sum_{k=1
}^j \int_\Omega \bm f\tk\cdot(\bm u\tk-\bm u\tkk)\mathop{{\rm d}x}.
\end{multline}
Next, we have 
\begin{equation}\label{eq:54}
\begin{aligned}
  &\sum_{k=1}^j\int_\Omega \bm f\tk\cdot(\bm u\tk-\bm u\tkk)\mathop{{\rm d}x}
\\
=&-\sum_{k=1}^{j-1}\int_\Omega(\bm f_\tau^{k+1}-\bm f_\tau^k)\cdot\bm u\tk\mathop{{\rm d}x}+\int_\Omega \big(\bm f_\tau^j\cdot\bm u_\tau^j-\bm f_\tau^1\bm u_\tau^0\big)\mathop{{\rm d}x},
\\
\le
&\sum_{k=1}^{j-1}\Big(
 \frac 1{2\tau} \|\bm f_\tau^{k+1}-\bm f\tk\|^2_{L^2}+\frac \tau 2 \|\bfu\tk\|^2_{L^2} \Big)\\
&\qquad +\frac\delta 2\|\bm u_\tau^j\|^2_{L^2}+\frac 1 {2\delta} \|\bm f_\tau^j\|_{L^2}^2+\|\bm f_\tau^1\|_{L^2}\|\bm u_\tau^0\|_{L^2},
\end{aligned}
\end{equation}
\color{black}for all $\delta>0$.\color{black}

Now, from \eqref{eq:54}, by making use of Lemma \ref{lem:11} and by \color{black}recalling that 
$\bm u_\tau^0=\bm u^0$  \color{black}we arrive at
\begin{equation}\label{eq:59}
\begin{aligned}
  &\sum_{k=1}^j\int_\Omega \bm f\tk\cdot(\bm u\tk-\bm u\tkk)\mathop{{\rm d}x}\le\sum_{k=1}^{j-1}\frac \tau 2 \|\bfu\tk\|^2_{L^2}
+\frac\delta 2\|\bm u_\tau^j\|^2_{L^2}+C_\delta.
\end{aligned}
\end{equation}
On combining \eqref{eq:8} with \eqref{eq:59}, and on using subsequently Korn's, Holder's, Young's, and Poincar\'e's inequalities, on choosing $\delta$ sufficiently small, we obtain
\begin{equation}\label{eq:60}
  \|\bfu_\tau^j\|_{H^1}^2+\tau\sum_{k=1}^{j}\|p\tk\|_{H^1}^2\le C\left(1+\tau\sum_{k=1}^{j-1}\|\bfu\tk\|^2_{H^1}\right).
\end{equation}
We now use the discrete Gronwall inequality in Lemma \ref{lem:1} with $a_0=C$ and $b_j=C\tau$ to obtain the following estimate:
\begin{align}\label{eq:14}
 \|\bm u_\tau^j\|^2_{H^1}\le Ce^{C\tau(j-1)}\le C e^{CT},
\end{align}
whence (\ref{eq:15}a). By \eqref{eq:14} the right-hand side of \eqref{eq:60} is bounded uniformly with respect to $\tau$, and we conclude that 
\begin{align}
  \int_0^T \|\overline p_\tau(t)\|_{H^1}^2\mathop{{\rm d}t}=\tau\sum_{k=1}^{n}\|p\tk\|_{H^1}^2\le C,
\end{align}
which entails (\ref{eq:15}b).
\end{proof}
\color{black}
\begin{proposition}[Time regularity]\label{prop:4}
The estimates
\begin{subequations}
\begin{equation}\label{eq:65a}
  \|\dot{\bm u}_\tau\|^2_{\color{black}L^2(0,T;H^1)}+\|\nabla \overline p_\tau\|^2_{\color{black}L^\infty(0,T;L^2)}\le C,
\end{equation}
\begin{equation}\label{eq:66a}
\|\dot{p}_\tau\|^2_{\color{black}L^2(0,T;L^2)}\le C,
\end{equation}
\end{subequations}
hold uniformly with respect to $\tau$.
\end{proposition}
\begin{proof}
\color{gt}
Thanks to the definition of $\bm u_\tau^0$ and $p_\tau^0$ given in \eqref{eq:6}, we have that $\bm u_\tau^0$ satisfies the weak form of the discrete force balance \eqref{eq:4a} and the boundary condition \eqref{eq:5a} for $k=0$. Thus, for all $k=1\dots n$ we have:\color{black}
\begin{equation}\label{eq:13}
\left\{ 
\begin{aligned}
&\int_\Omega \left(\mathbb C\frac{\bfE\bfu\tk-\bfE\bfu_\tau^{k-1}}\tau\cdot\bfE\bm v-\frac{p\tk-p_\tau^{k-1}}\tau\mathop{\rm div}\bm v\right){\rm d}x=\int_\Omega \frac{\bm f\tk-\bm f_\tau^{k-1}}\tau\cdot\bm v\mathop{{\rm d}x},\\
&\int_\Omega \left(\mathop{\rm div}\bfu\tk\, q+\tau\bm K\nabla p_\tau^k\cdot \nabla q\right){\rm d}x=\int_\Omega \mathop{\rm div}\bm u_\tau^{k-1}\,q\mathop{{\rm d}x},\\
&\textrm{for all $\bm v\in H^1_{u,D}(\Omega;\mathbb R^3)$ and $q\in H^1_{p,D}(\Omega)$.}
\end{aligned}\right.  
\end{equation}
On choosing $\bm v=(\bm u_\tau^k-\bm u_\tau^{k-1})/\tau$ and $q=(p_\tau^k-p_\tau^{k-1})/{\tau^2}$ as test functions in \eqref{eq:13} (the latter choice is legal also for $k=1$, since $p_0\in H^1_{p,D}(\Omega)$) and on adding the resulting equations we obtain, after making use of some elementary inequalities,
\begin{multline}\label{eq:11}
  c\left\|\frac{\bm u_\tau^k-\bm u_\tau^{k-1}}\tau\right\|^2_{H^1}+\frac c \tau \left({\left\|\nabla p\tk\right\|^2_{L^2}}-\left\|\nabla p_\tau^{k-1}\right\|^2_{L^2}\right)\\
\le \frac 1 {2\delta}\left\|\frac{\bm f\tk-\bm f_\tau^{k-1}}\tau\right\|^2_{L^2}+\frac \delta 2\left\|\frac{\bm u_\tau^k-\bm u_\tau^{k-1}}\tau\right\|_{L^2}^2,
\end{multline}
where $c>0$ depends on the constant of Korn's inequality and on the elasticity tensor $\mathbb C$. On choosing  $\delta=c$ we deduce, 
\begin{multline}\label{eq:12}
 \frac c 2\sum_{k=1}^j\tau\left\|\frac{\bm u_\tau^k-\bm u_\tau^{k-1}}\tau\right\|^2_{H^1}+\left\|\nabla p_\tau^j\right\|^2_{L^2}\\
\le \left\|\color{black}\nabla p_0\color{black}\right\|^2_{L^2} + \frac 1 {2 c }\sum_{k=1}^j\tau\left\|\frac{\bm f\tk-\bm f_\tau^{k-1}}\tau\right\|^2_{L^2},
\end{multline}
for every $j\in\{1,2,\ldots,N\}$.
The bound \eqref{eq:65a} is then obtained by invoking Lemma \ref{lem:11}, 
on noting that 
$$
\int_0^T \left\|\dot{\bm u}_\tau(t)\right\|^2_{H^1}{\rm d}t=\sum_{k=1}^n\tau\left\|\frac{\bm u_\tau^k-\bm u_\tau^{k-1}}\tau\right\|^2_{H^1}.
$$
We now prove \eqref{eq:66a}. From the first equation of \eqref{eq:13}, namely
$$
\int_\Omega \left(\mathbb C\frac{\bfE\bfu\tk-\bfE\bfu_\tau^{k-1}}\tau\cdot\bfE\bm v-\frac{p\tk-p_\tau^{k-1}}\tau\mathop{\rm div}\bm v\right){\rm d}x=\int_\Omega \frac{\bm f\tk-\bm f_\tau^{k-1}}\tau\cdot\bm v\mathop{{\rm d}x},
$$
which holds for all $\bm v\in H^1_{u,D}(\Omega;\mathbb R^3)$, we find
\begin{align}
\int_\Omega\frac{p\tk-p_\tau^{k-1}}\tau\mathop{\rm div}\bm v{\rm d}x&\le
C \left\|\frac{\bfE\bfu\tk-\bfE\bfu_\tau^{k-1}}\tau\right\|_{L^2}\|\bfE\bm v\|_{L^2}+
\left\|\frac{\bm f\tk-\bm f_\tau^{k-1}}\tau\right\|_{L^2}\|\bm v\|_{L^2}\nonumber\\
& \le
C \Big(\left\|\frac{\bfu\tk-\bfu_\tau^{k-1}}\tau\right\|_{H^1}+
\left\|\frac{\bm f\tk-\bm f_\tau^{k-1}}\tau\right\|_{L^2}\Big)\|\bm v\|_{H^1}.\label{eq:69}
\end{align}
Let $\varphi\in L^2(\Omega)$ with $\|\varphi\|_{L^2}\le 1$. Then by \cite[Lemma 3.2]{BauerNPS2015ECOaCoV}, there exists
a $\bm v_\varphi\in H^1_{u,D}(\Omega;\mathbb R^3)$ such that
$\mathop{\rm div} \bm v_\varphi=\varphi$ and $\|\bm v_\varphi\|_{H^1}\le C \|\varphi\|_{L^2}\le C$.
By taking $\bm v=\bm v_\varphi$ in \eqref{eq:69} we are led to
$$
\int_\Omega\frac{p\tk-p_\tau^{k-1}}\tau\varphi{\rm d}x\le C \Big(\left\|\frac{\bm u\tk-\bm u_\tau^{k-1}}\tau\right\|_{H^1}+
\left\|\frac{\bm f\tk-\bm f_\tau^{k-1}}\tau\right\|_{L^2}\Big),
$$
which holds for every $\varphi\in L^2(\Omega)$ with $\|\varphi\|_{L^2}\le 1$.
Thus
$$
\left\|\frac{p\tk-p_\tau^{k-1}}\tau\right\|_{L^2}\le C \Big(\left\|\frac{\bm u\tk-\bm u_\tau^{k-1}}\tau\right\|_{H^1}+
\left\|\frac{\bm f\tk-\bm f_\tau^{k-1}}\tau\right\|_{L^2}\Big).
$$
Taking squares, multiplying by $\tau$, and summing over $k$ we find
$$ 
\sum_{k=1}^N\tau \left\|\frac{p\tk-p_\tau^{k-1}}\tau\right\|_{L^2}^2\le C \Big(
 \|\dot{\bm u}_\tau\|^2_{L^2(0,T;H^1)} +
   \sum_{k=1}^n\tau\left\|\frac{\bm f\tk-\bm f_\tau^{k-1}}\tau\right\|^2_{L^2}\Big),
$$
and taking into account Lemma \ref{lem:11} \color{black}and 
\eqref{eq:65a} \color{black}we obtain \eqref{eq:66a}.
\end{proof}

\noindent\subsection{Passage to the limit.}

\begin{proposition}[Converging subsequences]\label{prop:5}
There exist
$$
\bm u\in H^1(0,T;H^1_{u,D}(\Omega;\mathbb R^3))
$$
and
$$
p\in L^\infty(0,T;H^1_{p,D}(\Omega))\cap H^1(0,T;L^2(\Omega))
$$ 
with $p(0)=p_0$ \color{black} and $\bm u(0)=\bm u_0$ where $\bm u_0$ is defined in \eqref{eq:28}, \color{black}such that, up to subsequences,
$$
\begin{array}{ll}
  \overline\bfu_\tau\stackrel{*}{\rightharpoonup} \bfu\textrm{ in }L^{\infty}(0,T;H^1_{u,D}(\Omega;\mathbb R^3)) & \quad\bfu_\tau{\rightharpoonup} \bfu\textrm{ in }H^1(0,T;H^1_{u,D}(\Omega;\mathbb R^3)),\\[1ex]
 \overline p_\tau\rightharpoonup p\textrm{ in }L^2(0,T;H^1_{p,D}(\Omega)),&\quad p_\tau\rightharpoonup p\textrm{ in }H^1(0,T;L^2(\Omega)).
\end{array}
$$
\end{proposition}
\begin{proof}
We start by proving that
\begin{equation}\label{pbarp}
\bar p_\tau - p_\tau \to 0, \quad\mbox{in }L^2(0,T;L^2(\Omega)).
\end{equation}
Indeed, using the definitions \eqref{eq:30}, we have that
\begin{align*}
\|\bar p_\tau - p_\tau\|_{L^2(0,T;L^2)}^2&=\sum_{k=1}^n \int_{(k-1)\tau}^{k\tau} \|\bar p_\tau(t) - p_\tau(t)\|_{L^2}^2{\rm d}t\\
&=\sum_{k=1}^n \int_{(k-1)\tau}^{k\tau} \|\color{black}p_\tau^k\color{black} - p_\tau^{k-1} \|_{L^2}^2 \Big(\frac{k\tau-t}{\tau}\Big)^2{\rm d}t\\
&\le\sum_{k=1}^n \tau \|\color{black}p_\tau^k \color{black}- p_\tau^{k-1} \|_{L^2}^2\\
&=\tau^2 \|\dot p_\tau\|_{L^2(0,T;L^2)}^2\le C\tau^2,
\end{align*}
where the last estimate follows from \eqref{eq:66a}.
Similarly we prove that
\begin{equation}\label{ubaru}
\bar {\bm u}_\tau - \bm u_\tau \to 0, \quad\mbox{in }L^2(0,T;H^1(\Omega)).
\end{equation}
Thanks to Proposition \ref{prop:2} there exist two subsequences  (which we do not relabel) such that
$$
\begin{aligned}
  &\overline\bfu_\tau\stackrel{*}{\rightharpoonup} \bfu\textrm{ in }L^{\infty}(0,T;H^1_{u,D}(\Omega;\mathbb R^3)),\\
 &\overline p_\tau\rightharpoonup p\textrm{ in }L^2(0,T;H^1_{p,D}(\Omega)).
\end{aligned}
$$
By Proposition \ref{prop:4} we also have that
$$
\nabla\overline p_\tau\stackrel{*}\rightharpoonup \nabla p\textrm{ in }L^\infty(0,T;L^2(\Omega;\mathbb R^3))
$$
and hence, by Poincar\'e inequality,  $p\in L^\infty(0,T;H^1_{p,D}(\Omega))$.
Also, by \eqref{pbarp}, \eqref{ubaru}, and Proposition \ref{prop:4} we deduce that
$$
\begin{aligned}
  &\bfu_\tau{\rightharpoonup} \bfu\textrm{ in }H^1(0,T;H^1_{u,D}(\Omega;\mathbb R^3),\\
 &p_\tau\rightharpoonup p\textrm{ in }H^1(0,T;L^2(\Omega)).
\end{aligned}
$$
Finally, since $p_\tau(0)=\color{black}p_0$ \color{black} and $\bm u_\tau(0)=\bm u_0$ \color{black}we have, from the convergence above, that $p(0)=\color{black}p_0$ \color{black}and $\bm u(0)=\bm u_0$, \color{black} as required.
\end{proof}

\subsection{Proof of Theorem \ref{thm:main}}
Let $\overline{\bm f}_\tau$ be the piecewise constant interpolation of $\bm f\tk$, as in \eqref{eq:30}. 
From assumption \eqref{eq:31c} and definition \eqref{eq:32} it follows that
$$
  \overline{\bm f}_\tau\to\bm f\quad\textrm{ in } L^2(0,T;L^2(\Omega)).
$$
Choose any test function $\varphi\in\mathcal D(0,T)$. From the first of \eqref{eq:50} we obtain 
$$
   \int_0^T\int_{\Omega}\big(\mathbb C\bm E\overline\bfu_\tau(t)\cdot\bm E\bm v-\overline p_\tau(t)\mathop{\rm div}\bm v)\mathop{{\rm d}x}\varphi(t)\mathop{{\rm d}t}=\int_0^T\int_{\Omega}\overline{\bm f}_\tau(t)\cdot\bm v\mathop{{\rm d}x}\varphi(t)\mathop{{\rm d}t},
$$
where $\bm v$ is an arbitrary function in $H^1_{{u,D}}(\Omega;\mathbb R^3)$.
Recalling Proposition \ref{prop:5} and on passing to the limit we obtain
$$
   \int_0^T\left(\int_{\Omega}\big(\mathbb C\bm E\bfu(t)\cdot\bm E\bm v- p(t)\mathop{\rm div}\bm v)\mathop{{\rm d}x}\right)\varphi(t)\mathop{{\rm d}t}=\int_0^T\left(\int_{\Omega}{\bm f}(t)\cdot\bm v\mathop{{\rm d}x}\right)\varphi(t)\mathop{{\rm d}t},
$$
\color{black}whence, by the arbitrariness of $\varphi$ we have, for a.e. $t\in(0,T)$,
\begin{equation}\label{eq:92}
   \int_{\Omega}\big(\mathbb C\bm E\bfu(t)\cdot\bm E\bm v- p(t)\mathop{\rm div}\bm v)\mathop{{\rm d}x}=\int_{\Omega}{\bm f}(t)\cdot\bm v\mathop{{\rm d}x}.
\end{equation}
\color{black}Since $\bm u(0)=\bm u_0$ with $\bm u_0$ solving \eqref{eq:28}, the equation \eqref{eq:92} holds also at $t=0$. Moreover, the continuity of $\bm u$, $p$, and $\bm f$ implies that \eqref{eq:92} holds for all $t\in[0,T)$.\color{black}

The second of \eqref{eq:50} yields
$$
\int_0^T\int_{\Omega}\big(\mathop{\rm div}\dot\bfu_\tau q+\bm K\nabla \overline p_\tau\cdot\nabla q\big)\mathop{{\rm d}x}\varphi(t)\mathop{{\rm d}t}=0.
$$
Again, recalling Proposition \ref{prop:5} and passing to the limit we obtain 
$$
\int_0^T\int_{\Omega}\big(\mathop{\rm div}\dot\bfu \,q+\bm K\nabla p\cdot\nabla q\big)\mathop{{\rm d}x}\varphi(t)\mathop{{\rm d}t}=0,
$$
\color{black}and by localizing we obtain
$$
\int_{\Omega}\big(\mathop{\rm div}\dot\bfu \,q+\bm K\nabla p\cdot\nabla q\big)\mathop{{\rm d}x}=0\quad\textrm{for a.e. $t\in(0,T)$.}
$$
We next address the uniqueness of the solution. We choose $\bm v=\dot{\bm u}(t)$ and $q=p(t)$ in the weak formulation \eqref{eq:27},  we add the resulting equations and for $s\in (0,T]$ we integrate by parts over the time interval $(0,s)$ to obtain
\begin{equation}\label{eq:24}
\begin{aligned}
\frac 12\int_\Omega \mathbb C\bm E{\bm u}(s)&\cdot \bm E\bm u(s){\rm d}x+\int_0^s\int_\Omega \bm K\nabla p(t)\cdot\nabla p(t){\rm d}x{\rm d}t
\\
&=\frac 12 \int_\Omega \mathbb C\bm E\color{black}\bm u(0)\color{black}\cdot\bm E\color{black}\bm u(0)\color{black}{\rm d}x+\int_0^s \int_\Omega \bm f(t)\cdot\dot{\bm u}(t){\rm d}x{\rm d}t.
\end{aligned}
\end{equation}
A further integration by parts on the right--hand side of \eqref{eq:24} yields
\begin{equation}\label{eq:29}
\begin{aligned}
\mathcal E(s)+\int_0^s\int_\Omega \bm K\nabla p(t)\cdot\nabla p(t){\rm d}x{\rm d}t=\mathcal E(0)-\int_0^s \int_\Omega \dot{\bm f}(t)\cdot{\bm u}(t){\rm d}x{\rm d}t,
\end{aligned}
\end{equation}
where 
\begin{equation}\label{eq:33}
\mathcal E(s):=\int_\Omega \left(\frac 12 \mathbb C\bm E\bm u(s)\cdot\bm E\bm u(s)-\bm f(s)\bm u(s)\right){\rm d} x,
\end{equation}

The coercivity of the elasticity tensor $\mathbb C$ assumed in \eqref{eq:31d}, Young's inequality, and Korn's inequality yield
\[
\begin{aligned}
\mathcal E(s)&\ge \frac  {c_{\mathbb C}}2\|\bm E\bm u(s)\|^2_{L^2}-\frac 1 \delta \|\bm f(s)\|^2_{L^2}-\delta \|\bm u(s)\|^2_{L^2}\\
&\ge c_1\|\bm u(s)\|_{H^1}^2-\frac {C_1}{\delta}\left(\|\bm f(0)\|^2_{L^2}+\|\bm f\|^2_{H^1(0,T;L^2)}\right)-\delta \|\bm u(s)\|^2_{H^1},
\end{aligned}
\]
where $c_1:=c_Kc_{\mathbb C}/{2}$, with \color{black}$c_K:=\inf\left\{\|\bm E\bm u\|_{L^2}/{\|\bm u\|_{H^1},\bm u\neq \bm 0}\right\}>0$ \color{black}and $C_1=2(1+T)$. \color{black} For $\delta=c_1/2$ the above chain \blue of inequalities \black yields
\begin{equation}\label{eq:41}
\|\bm u(s)\|^2_{H^1}\le C_2\left(\mathcal E(s)+\|\bm f(0)\|^2_{L^2}+\|\bm f\|^2_{H^1(0,T;L^2)}\right),
\end{equation}
where $C_2=\max\left(\frac 2 {c_1},4\frac {C_1}{c_1^2}\right)$. Moreover, it follows from Poincar\'e's inequality that
\begin{equation}\label{eq:68}
\int_0^s\|p(t)\|^2_{H^1}{\rm d}t\le C_3 \int_0^s\int_\Omega \bm K\nabla p(t)\cdot\nabla p(t){\rm d}x{\rm d}t,
\end{equation}
for a sufficiently large constant $C_3$. It follows from \eqref{eq:29}, \eqref{eq:41}, and \eqref{eq:68} that, for \color{black}$C_4=\max(C_2,C_3)$,\color{black}
\begin{equation}\label{eq:44}
\begin{aligned}
\|\bm u(s)\|_{H^1}^2+&\int_0^s\|p(t)\|^2_{H^1}{\rm d}t\nonumber\\
&\le C_4\left(\blue \mathcal E(0)\black-\int_0^s \dot{\bm f}(t)\cdot\bm u(t){\rm d}t\right)+C_2\left(\|\bm f(0)\|^2_{L^2}+\|\bm f\|^2_{H^1(0,T;L^2)}\right)\\
&\le \color{black}C_5\color{black}\left(\color{black}\mathcal E(0)\black+\int_0^s \|\bm u(t)\|^2_{H^1}{\operatorname d}t \blue+ \|\bm f(0)\|^2_{L^2}+\|\bm f\|^2_{H^1(0,T;L^2)}\right),
\end{aligned}
\end{equation}
where $C_5$ is a suitably large constant. \color{black}
Consider now two solutions, say $(\bm u_1,p_1)$ and $(\bm u_2,p_2)$, and let $\blue\delta {\bm u}=\bm u_1-\bm u_2$ and $\blue\delta  p=p_1-p_2$. The pair $(\blue\delta {\bm u},\blue\delta  p)$ is a solution of \eqref{eq:27} with homogeneous forcing $\bm f=\bm 0$, and with homogeneous initial condition $\blue\delta  p(0)=0$. \color{black}The uniqueness of the solution of the elasticity problem implies that $\blue\delta {\bm u}(0)=\bm 0$. \color{black}Thus, for such solution, the estimate \eqref{eq:44} holds with $\color{black}\mathcal E(0)=0$ and $\bm f=\bm 0$, namely, 
\begin{equation}\label{eq:45}
\begin{aligned}
\|\blue\delta {\bm u}(s)\|_{H^1}^2+\int_0^s\|\blue\delta  p(t)\|^2_{H^1}{\rm d}t\le \frac {C_4}2\int_0^s \|\blue\delta {\bm u}(t)\|^2_{H^1}{\rm d}t.
\end{aligned}
\end{equation}
From \eqref{eq:45} and Gronwall's inequality we obtain that $\blue\delta {\bm u}=0$ and $\blue\delta  p=0$, thus $\bm u_1=\bm u_2$ and $p_1=p_2$, as required.
\medskip

\noindent{\bf Remark. Taking into account the  applied pressure $p_{\rm a}$.} 
If the simplifying assumption \eqref{eq:48} is removed, the existence and uniqueness of a weak solution to \eqref{eq:1}--\eqref{eq:3} can be proved with little conceptual difficulty, provided that the time--dependent pressure field $p_{\rm a}(t)$ prescribed on $\Gamma_{p,D}$ satisfies
\begin{equation}\label{eq:70}
p_{\rm a}\in H^1(0,T;H^{1/2}(\Gamma_{p,D})).
\end{equation}
Indeed, let $\tilde p_{\rm a}(t)$ be the lifting of $p_{\rm a}(t)$ to $\Omega$ defined by
\begin{equation}\label{eq:56}
\begin{aligned}
{\rm div}\bm K\nabla\tilde p_{\rm a}(t)&=0&&\mbox{\rm in}\quad\Omega,\\
\tilde p_{\rm a}(t)&=p_{\rm a}(t)&&\mbox{\rm on}\quad\Gamma_{p,D},\\
-\bm K\nabla\tilde p_{\rm a}(t)\cdot{\bm n}&=0&&\mbox{\rm on}\quad\Gamma_{p,N}.
\end{aligned}
\end{equation}
Problem \eqref{eq:1}--\eqref{eq:3} can be reformulated  in terms of the unknowns $\bm u(t)$ and
\[
\blue\widehat p\black(t):=p(t)-\tilde p_{\rm a}(t)
\]
to obtain the following system:
\begin{subequations}\label{eq:1bis}
 \begin{align}
  &\mathop{\rm div}\mathbb C\bm E\bm  u-\nabla \color{black}\widehat p=-\bm f+\nabla\widetilde p_{\rm a},
\\
  &\mathop{\rm div}\dot\bfu-\mathop{\rm div}\bm K\nabla \widehat p=0,
\end{align}
\end{subequations}
with boundary conditions:
\begin{subequations}\label{eq:2bis}
\begin{align}
\bfu(t)&=\bm 0\textrm{ on }\Gamma_{u,D},&(\mathbb C\bm E\bm u(t))\bm n-\color{black}\widehat p(t)\bm n&=\tilde{p}_{\rm a}(t)\bm n\textrm{ on }\Gamma_{u,N}, \label{eq:2a}\\
\color{black}\widehat p(t)&=0\textrm{ on }\Gamma_{p,D},&-\bm K\nabla \color{black}\widehat p(t)\cdot\bm n&=\bm 0\textrm{ on }\Gamma_{p,N}
,\label{eq:2b}
\end{align}
\end{subequations}
and with the initial condition
\begin{align}\label{eq:3bis}
  \color{black}\widehat p(0)=\color{black}\widehat p_0:=p_0-\widetilde p_{\rm a}.
\end{align}
The weak formulation of \eqref{eq:1bis}--\eqref{eq:3bis} is \
\begin{equation}\label{eq:27}
\left\{ 
\begin{aligned}
  &\bfu\in H^1(0,T;H^1_{u,D}(\Omega;\mathbb R^3)),\\
&\color{black}\widehat p\in L^\infty(0,T;H^1_{p,D}(\Omega))\cap H^1(0,T;L^2(\Omega))\quad\text{ with } \quad \color{black}\widehat p(0)=\color{black}\color{black}\widehat p_0,\\
&\int_{\Omega}\big(\mathbb C\bm E\bfu(t)\cdot\bm E\bm v- p(t)\mathop{\rm div}\bm v)\mathop{{\rm d}x}=\langle \bm\ell(t),\bm v\rangle,\\
&\qquad\forall\bm v\in H^1_{u,D}(\Omega;\mathbb R^3),\color{black}\forall t\in(0,T),\\
&\int_\Omega \left(\dot\bfu(t)q+\bm K\nabla p(t)\cdot\nabla q\right)\mathop{{\rm d}x}=0\quad \forall q\in H^1_{p,D}(\Omega),\text{ for a.e. }t\in(0,T).\\
\end{aligned}\right. 
\end{equation}
where the time--dependent linear functional $\bm\ell\in H^1(0,T;(H^1_{u,D}(\Omega;\mathbb R^3))')$ is defined by
\[
\langle \bm\ell(t),\bm v\rangle:=\int_\Omega \left(\bm f(t)\cdot\bm v+\tilde p_{\rm a}{\rm div}\bm v\right){\rm d}x.
\]
The proof of existence and uniqueness requires only minor conceptual changes. Namely, every occurrence of a scalar product $(\bm f(t),\bm v)=\int_\Omega \bm f(t)\cdot\bm v{\rm d}x$ must be replaced with the pairing $\langle \bm l(t),\bm v\rangle$, and every occurrence of the Cauchy--Schwartz inequality $(\bm f(t),\bm v)\le \|\bm f(t)\|_{L^2}\|\bm v\|_{L^2}$ must be replaced with $\langle\bm\ell (t),\bm v\rangle\le \|\bm f(t)\|_{(H^1_{u,D})'}\|\bm v\|_{H^1_{u,D}}$.
\color{black}

\section{Plate equations}\label{sec:plate-equations}
\paragraph{A family of shrinking plate--like domains.} \blue For $\varepsilon>0$ a parameter that tends to null, 
\black we consider a slab of thickness $\varepsilon h$ modeled on a plane domain $\omega$:
\begin{equation}\label{eq:1032}
\Omega_\eps=\omega\times\left(-\eps \frac h 2,+\eps \frac h 2\right),\qquad \omega\subset\mathbb R^2.
\end{equation}
We fix a part $\gamma_{u,D}\subset\partial\blue\omega$, we prescribe null displacement on
\begin{equation}\label{eq:1033}
\Gamma_{u,D,\eps}=\gamma_{u,D}\times\left(-\eps\frac h2,+\eps\frac h2\right),
\end{equation}
and \blue null traction \black on 
\[
\Gamma_{u,N,\eps}=\partial\Omega_\eps\setminus\Gamma_{u,D,\eps}.
\]
On the top and bottom faces of the slab 
\begin{equation}
  \Gamma_{p,D,\eps}=\omega\times\left\{\pm\eps \frac h 2\right\}
\end{equation}
we impose a pressure field $p_{\rm a,\eps}$. On the lateral side of the plate
\begin{equation}
\Gamma_{p,N,\eps}=\partial\omega\times\left(-\eps\frac h2,+\eps\frac h2\right)
\end{equation}
we require that the flux be null.

\paragraph{A family of evolution problems.}
For each $\eps$ we consider the following system 
\begin{subequations}\label{eq:52}
 \begin{align}
  &\mathop{\rm div}\mathbb C_\eps\bm E\bm  u_\eps-\nabla p_\eps=-\bm f_\eps,
\\
  &\mathop{\rm div}\dot\bfu_\eps-\mathop{\rm div}\bm K_\eps\nabla p_\eps=0,
\end{align}
\end{subequations}
in the unknowns $\blue\bfu_\eps$ and $\blue p_\eps$, formulated in the space-time domain $\Omega_\eps\times(0,T)$, with boundary conditions
\begin{subequations}\label{eq:26}
\begin{align}
\bfu(t)&=\bm 0\textrm{ on }\Gamma_{u,D,\eps},&(\mathbb C_\eps\bm E\bm u_\eps(t))\bm n-p_\eps(t)\bm n&=0\textrm{ on }\Gamma_{u,N,\eps},\\
p_\eps(t)&=p_{\rm a,\eps}(t)\textrm{ on }\Gamma_{p,D,\eps},&-\bm K_\eps\nabla p_\eps(t)\cdot\bm n&=\bm 0\textrm{ on }\Gamma_{p,N,\eps}\blue,
\end{align}
\end{subequations}
and with the initial condition 
\begin{equation}\label{eq:93}
  \blue p_\eps(0)=p_{0,\eps}
\end{equation}
We make the following assumptions
\begin{subequations}\label{eq:10}
\begin{align}
&\mathcal H^1(\gamma_{u,D})>0\\
&p_{0,\eps}\in H^1_{p,D,\eps}(\Omega_\eps),\qquad p_{0,\eps}=p_{\rm a,\eps}\text{ on }\Gamma_{p,D,\eps},\\  
&\bm f_\eps\in H^1(0,T;L^2(\Omega_\eps;\mathbb R^3)),\\
&(\mathbb C_{\eps})_{ijkl}\in L^\infty(\Omega_\eps),\\
&(\mathbb C_{\eps})_{ijkl}=(\mathbb C_\eps)_{klij}=(\mathbb C_\eps)_{jikl},\qquad (\mathbb C_\eps)_{333\alpha}=(\mathbb C_\eps)_{3\alpha\beta\gamma}=0,\label{eq:10e}\\
&p_{\rm a,\eps}\in H^1(0,T;H^{1/2}(\Gamma_{p,D,\eps})),\\
&\text{$\exists c_{\mathbb C}>0:\mathbb C_\eps\bm A\cdot\bm A\ge c_{\mathbb C}|\bm A|^2\quad\forall \eps>0,\bm A\in\mathbb R^{3\times 3}_{\rm Sym}$}\quad \text{a.e. in }\Omega_\eps,\\
&\blue \mathbf K_\eps\in L^\infty(\Omega_\eps;\mathbb R^{3\times 3}_{\rm Sym}),\qquad (K_\varepsilon)_{3\alpha}=0.
\label{eq:10h}\\
&\text{$\exists c_{\bm K}>0:\mathbf K_\eps\bm a\cdot\bm a\ge c_{\bm K}|\bm a|^2\quad\forall\eps,\forall \bm a\in\mathbb R^3$}\quad \text{a.e. in }\Omega_\eps.\label{eq:10i}
\end{align}
\end{subequations}
As pointed out in the remark at the conclusion of the previous section, it is convenient to decompose the pressure $p_\eps$ into the sum of 
\begin{itemize}
\item a fluctuating component $\color{black}\widehat p_{\eps}$, that is coupled with the displacement $\bm u$ and 
\item a component $\tilde p_{\rm a,\eps}$ that is directly controlled through the applied pressure $p_{\rm a,\eps}$.
\end{itemize}
To this effect, we set
\[
p_\eps(t)=\color{black}\widehat p_\eps(t)+\tilde p_{\rm a,\eps}(t),
\]
where $\tilde p_{\rm a,\eps}$ is the lifting of the boundary datum to $\Omega_\eps$ obtained by solving the following (time--dependent) boundary value problem:
\begin{equation}\label{eq:76}
\begin{aligned}
{\rm div}\bm K_{\eps}\nabla\tilde p_{\rm a,\eps}(t)&=0&&\mbox{\rm in}\quad\Omega_\eps,\\
\tilde p_{\rm a,\eps}(t)&=p_{\rm a,\eps}(t)&&\mbox{\rm on}\quad\Gamma_{p,D,\eps},\\
-\bm K_\eps\nabla\tilde p_{\rm a,\eps}(t)\cdot{\bm n}&=0&&\mbox{\rm on}\quad\Gamma_{p,N,\eps}.
\end{aligned}
\end{equation}
We introduce the notation $H^1_{p,D,\eps}(\Omega_\varepsilon)=\left\{q\in H^1(\Omega_\eps):q=0\text{ on }\Gamma_{p,D,\eps}\right\}$. \blue At a given $t\in[0,T)$, \black the weak form of \eqref{eq:76} is
\begin{equation}\label{eq:79}
\begin{cases}
&p_{\eps,\rm a}(t)\in H^1(\Omega_\eps)\qquad\text{with}\qquad \tilde p^\eps_{\rm a}(t)=p_{\rm a}^\eps(t)\text{ on }\Gamma_{p,D,\eps},\\[0.5em]
&\displaystyle\int_{\blue\Omega_\eps} \bm K_\eps\nabla \tilde p_{\rm a,\eps}(t)\cdot\nabla q\mathop{{\rm d}x}=0\quad\forall q\in H^1_{p,D,\eps}(\Omega_\eps).
\end{cases}
\end{equation}\color{black}
The weak formulation of the problem governing $(\bm u_\eps,\color{black}\widehat p_\eps)$ is the following:
\begin{equation}\label{eq:53}
\left\{ 
\begin{aligned}
  &\bfu_\varepsilon\in H^1(0,T;H^1_{u,D,\eps}(\Omega_\varepsilon;\mathbb R^3)),\\
&\color{black}\widehat p_\varepsilon\in L^\infty(0,T;H^1_{p,D,\eps}(\Omega_\varepsilon))\cap H^1(0,T;L^2(\Omega_\eps))\quad\text{ with } \quad \color{black}\widehat p_\eps(0)=\color{black}\color{black}\widehat p_{0,\eps},\comment{Controllare se serve che il dato iniziale sia ben preparato}\\
&\int_{\Omega_\eps}\big(\mathbb C_\eps\bm E\bfu_\eps(t)\cdot\bm E\bm v-\color{black}\widehat p_\eps(t)\mathop{\rm div}\bm v)\mathop{{\rm d}x}=\int_{\Omega_\eps}\big({\bm f}_\eps(t)\cdot\bm v+\tilde p_{\rm a,\eps}\mathop{\rm div}\bm v\big)\mathop{{\rm d}x}\\
&\qquad\forall\bm v\in H^1_{u,D,\eps}(\Omega_\eps;\mathbb R^3),\color{black}\forall t\in[0,T\color{black})\color{black},\\
&\int_{\Omega_\eps} {\rm div}\dot\bfu_\eps(t)q+\bm K_\eps\nabla \color{black}\widehat p_\eps(t)\cdot\nabla q\mathop{{\rm d}x}=0\qquad \forall q\in H^1(\Omega_\eps),\text{ for a.e. }t\in(0,T).\\
\end{aligned}\right.
\end{equation}
Here we use the notation $H^1_{u,D,\eps}(\Omega_\varepsilon;\mathbb R^3)=\left\{\bv\in H^1(\Omega_\eps;\mathbb R^3):\bv=\bm 0\text{ on }\Gamma_{u,D,\eps}\right\}$. Moreover,
\[
\color{black}\widehat p_{0,\eps}=p_{0,\eps}-\tilde p_{\rm a,\eps}(0).
\]
The existence and uniqueness of a weak solution follows from Theorem \ref{thm:main} and from the remarks at the end of the previous section.

\paragraph{Change of independent variables. } We next reformulate the problem on a domain that does not depend on $\eps$. An obvious choice is  
\[\Omega=\omega\times\left(-\frac h 2,+\frac h 2\right),
\] 
which corresponds to taking $\eps=1$ in \eqref{eq:1032}. To this effect, we introduce the linear map 
\[
\blue
\bm r_\eps:\Omega_\eps\to\Omega,\qquad \bm r_\eps(\bm x):=\bm R_\eps\bm x,\qquad \bm R_\eps:=\mathop{\rm diag}(1,1,\eps^{-1})
\]
(\emph{i.e.} $\bm R^\eps$ is the diagonal matrix with entries $1$, $1$, and $\blue\eps^{-1}$). \blue We can therefore write the typical point $\bm x$ of the fixed domain $\Omega$ as
\[
\bm x=\bm R_\eps\bm x_\eps
\]
where $\bm x_\eps$ is a typical point of the shrinking domain $\Omega_\eps$.\black 

\paragraph{Change of unknowns.} At this stage, we express the fields $\bm u_\eps$ and $p_\eps$ in terms of $\bm x$ and write the corresponding evolution problem on the space-time domain $\Omega\times (0,T)$. The result will be a singular-perturbation problem to study through asymptotic analysis. Before doing that, however, we go through an additional step: we change the dependent variables that so that it is easier to extract information from the resulting singular perturbation problem. Our choice of the new dependent variables is suggested by known results concerning the purely mechanical problem of bending of a linearly elastic plate. For this problem it is known \cite{Ciarl1997} that in the limit as $\eps$ tends to null the ratio between in--plane displacement and transverse displacement is of the order of $\varepsilon$. This motivates the introduction of the scaled displacement defined by:
\begin{equation}\label{eq:74}
u^\eps_\alpha(\bm x,t):=\frac 1 \eps({\bm u_\eps})_\alpha(\bm x^\eps,t)\qquad {u^\eps_3}(\bm x,t):=({\bm u_\eps})_3(\bm x^\eps,t)
\end{equation}
(henceforth Greek free indices range between 1 and 2). As to the pressure, \blue we select the following change of variable:\black
\begin{equation}\label{eq:75}
\color{black}\widehat p^\eps(\bm x,t)=\frac 1\eps\color{black}\widehat p_\eps(\bm x^\eps,t),\qquad \tilde p^\eps_{\rm a}(\bm x,t)=\frac 1 \eps\tilde p_{\rm a,\eps}(\bm x^\eps,t).
\end{equation}
\blue One checks that $\bm u^\eps(t)=\eps\bm R_{\eps}\bm u_\eps(\bm x,t)\circ{\bm r^\eps}$. This yields $\nabla\bm u^\eps(\bm x^\eps,t)=\eps{\bm R_{\eps}}\nabla\bm u_\eps(\bm x,t){\bm R_\eps}=\nabla^\eps\bm u^\eps$, that is
\begin{equation}\label{eq:77}
\nabla\bm u_\eps(\bm x^\eps,t)=\eps\nabla^\eps \bm u^\eps(\bm x,t)\qquad\text{where}\quad\nabla^\eps\bm u^\eps:=\bm R_\eps\nabla\bm u^\eps\bm R_\eps=
\begin{pmatrix}
\partial_{\alpha\beta}\bm u^\eps & \displaystyle \frac {\partial_\alpha u^\varepsilon_3} {\varepsilon} 
\\[0.8em]
 \displaystyle  \frac {\partial_\beta u^\varepsilon_3}{\varepsilon} &  \displaystyle  \frac  {\partial_3 u^\varepsilon_3} {\varepsilon^2}
\end{pmatrix},
\end{equation}
whence, in particular,
\begin{equation}
  \bm E\bm u^\eps(\bm x^\eps,t)=\eps\bm E^\eps\bm u_\eps(\bm x,t),\qquad\text{where}\quad \bm E^{\eps}\bm u_\eps:=\operatorname{sym}
\nabla^\eps\bm u_\eps,
\end{equation}
and\black
\begin{equation}
\operatorname{div}\bm u_\eps(\bm x^\eps,t)=\eps\operatorname{div}^\eps\bm u^\eps(\bm x,t)\qquad\text{where}\quad\operatorname{div}^\eps\bm u^\eps:=\operatorname{tr}\nabla^\eps\bm u^\eps.
\end{equation}
Similarly, the gradient of the fluctuating part $\color{black}\widehat p_\eps$ of the pressure can be rendered in terms of the spatial derivatives of new unknown $\color{black}\widehat p^\eps$ through
\begin{equation}\label{eq:78}
\nabla\blue\widehat p_\eps(\bm x^\eps,t)=\eps\nabla^\eps \color{black}\widehat p^\eps(\bm x,t)\qquad\text{where}\quad\nabla^\eps\color{black}\widehat p^\eps:=\begin{pmatrix}
{\blue\partial_\alpha\widehat p^\eps}
\\[0.5em]
\displaystyle\frac {\partial_3\color{black}\widehat p^\eps} \eps
\end{pmatrix},
\end{equation}
with a similar formula holding for $\nabla\tilde p_\eps$. 

\paragraph{The singular perturbation problem.}
We  introduce the shorthand notation $\Gamma_{u,D}\equiv\Gamma_{u,D,1}$, $\Gamma_{u,N}\equiv\Gamma_{u,N,1}$, $\Gamma_{p,D}\equiv\Gamma_{p,D,1}$, $\Gamma_{p,N}\equiv\Gamma_{p,N,1}$ and we define
\[
\mathbb C^\eps(\bm x):=\mathbb C_\eps(\bm x^\eps), \qquad \mathbf K^\eps(\bm x):=\mathbf K_\eps(\bm x^\eps),
\]
\[
f^\eps_\alpha(\bm x,t):=\eps(f_\eps)_\alpha(\bm x^\eps,t),\qquad f^\eps_3(\bm x,t):=\eps^2(f_\eps)_\alpha(\bm x^\eps,t).
\]
\[
p_{\rm a}^\eps(\bm x,t):=p_{\rm a,\eps}(\bm x^\eps,t).
\]
It is now easy to check that $\tilde p_{\rm a,\eps}(t)$ solves \blue\eqref{eq:75} \black if and only if $\tilde p_{\rm a}^\eps(t)$ solves: 
\begin{equation}\label{eq:58}
\blue
\begin{cases}
&\tilde p_{\rm a}^\eps\in H^1(\Omega), \qquad \tilde p^\eps_{\rm a}(t)=p_{\rm a}^\eps(t)\text{ on }\Gamma_{p,D},\\[0.5em]
&\displaystyle\int_\Omega \bm K^\eps\nabla^\eps \tilde p_{\rm a}^\eps(t)\cdot\nabla^\eps q\mathop{{\rm d}x}=0
\quad\forall q\in H^1_{p,D}(\Omega).
\end{cases}
\end{equation}
Moreover, $(\bm u_\eps,\color{black}\widehat p_\eps)$ solves Problem \eqref{eq:53} if and only if  $(\bm u^\eps,\color{black}\widehat p^\eps)$ solves the following problem:\color{black}
\begin{equation}\label{eq:9}
\left\{ 
\begin{aligned}
  &\bfu^\varepsilon\in H^1(0,T;H^1_{u,D}(\Omega;\mathbb R^3)),\\
&\color{black}\widehat p^\varepsilon\in L^\infty(0,T;H^1_{p,D}(\Omega))\cap H^1(0,T;L^2(\Omega))\quad\text{ with } \quad \color{black}\widehat p^\eps(0)=\color{black}\color{black}\widehat p^\eps_{0},\footnote{Controllare se \`e necessario che il dato sia ben preparato}\\
&\int_{\Omega}\big(\mathbb C^\eps\bm E^\eps\bfu^\eps(t)\cdot\bm E^\eps\bm v- \color{black}\widehat p^\eps(t){\mathop{\rm div}}^\eps\bm v)\mathop{{\rm d}x}=\int_{\Omega}\big({\bm f}^\eps(t)\cdot\bm v+\tilde p_{\rm a}^\eps{\mathop{{\rm div}^\eps}}\bm v\big)\mathop{{\rm d}x}\\
&\qquad\forall\bm v\in H^1_{u,D}(\Omega;\mathbb R^3),\color{black}\forall t\in[0,T\color{black})\color{black},\\
&\int_{\Omega} {\rm div}^\eps\dot\bfu^\eps(t)q+\bm K^\eps\nabla^\eps \color{black}\widehat p^\eps(t)\cdot\nabla^\eps q\mathop{{\rm d}x}=0\qquad\forall q\in H^1_{p,D}(\Omega),\text{ for a.e. }t\in(0,T).\\
\end{aligned}\right.
\end{equation} \color{black}

\color{black}

\paragraph{Further assumptions.}
We consider a regime when $\bm K^\eps$ \black tends to a limit $\bm K$ as $\eps\to 0$ in the following sense:
\begin{equation}\label{eq:85}
\bm K^\eps\to \bm K\qquad \text{in }L^\infty(\Omega;\mathbb R^{3\times 3}).
\end{equation}
We shall make the following assumptions concerning the other data:
\begin{subequations}\label{eq:63}
\begin{align}
&\mathbb C^\eps\to\mathbb C&&\text{in }L^\infty(\Omega;\mathbb R^{3\times 3\times 3\times 3}),\label{eq:63a}\\
&{\bm f}^\eps\to {\bm f}&&\text{in }H^1(0,T;L^2(\Omega)),\\
&p_{\rm a}^\eps\to p_{\rm a}&&\text{in }H^1(0,T;H^{1/2}(\Gamma_{p,D})). \label{eq:63d}
\end{align}
\end{subequations}
\medskip



\noindent{\bf Convergence of the pressure field.} Our first step is to study the asymptotic behavior of the lifting $\tilde p_{\rm a}^\eps$. We show that in the limit $\tilde p^\eps_{\rm a}$ approximates the linear interpolation $\hat p_{\rm a}^\eps$ of the pressures applied at the top and at the bottom of the plate. In what follows we set:
\[
p_{\rm a}^{\eps,\pm}(x_1,x_2,t):=p_{\rm a}^\eps(x_1,x_2,\pm h/2,t).
\]
Also, we use the same notation for the limit $p_{\rm a}$:
\[
p_{\rm a}^{\pm}(x_1,x_2,t):=p_{\rm a}(x_1,x_2,\pm h/2,t).
\]
We next define:
\[
  \zeta(x):=h\left(\frac{\mathcal K^{-1}_{33}(x)}{\mathcal K_{33}^{-1}(x_1,x_2,h/2)}-\frac 12\right)\qquad\text{where}\qquad\mathcal K_{33}^{-1}(x):=\int_{-h/2}^{x_3}\frac 1 {K_{33}(x_1,x_2,z)}{\rm d}z.
\]\color{black}
\begin{lemma}\label{lem:3}
Under assumptions \eqref{eq:85}, let $\tilde p_{\rm a}^\eps(t)\in H^1(\Omega)$ solve \eqref{eq:58} for all $t\in[0,T)$. Then
\[
\tilde p_{\rm a}^\eps\rightharpoonup \tilde p_{\rm a}\quad\textrm{in }H^1(0,T;L^2(\Omega)),\qquad \partial_3\tilde p_{\rm a}^\eps\rightharpoonup \partial_3\tilde p_{\rm a}\quad\textrm{in }H^1(0,T;L^2(\Omega)),
\]
where
\begin{equation}\label{eq:87}
 \tilde p_{\rm a}=\frac{p^+(x_1,x_2)+p^-(x_1,x_2)}2+\zeta(x)\frac{p^+(x_1,x_2)-p^-(x_1,x_2)}{h}.\color{black}
\end{equation}
\end{lemma}
\begin{remark}\label{rem:4}
If $K_{33}$ does not depend on $x_3$ then $\zeta(x)=x_3$ and the representation formula \eqref{eq:87} results in $\tilde p_{\rm a}$ being the affine interpolation between the pressures applied at the top and at the bottom of the plate. 
\end{remark}
\color{black}
\begin{proof}[Proof of Lemma \ref{lem:3}.]
 By standard results on liftings of traces, there exists $\check p_{\rm a}^\eps\in H^1(0,T;H^1(\Omega))$ such that $\check p_{\rm a}^\eps(t)=p_{\rm a}^\eps(t)$ on $\Gamma_{p,D}$ for all $t\in[0,T)$ and 
 \begin{equation}\label{eq:81}
\|\check p_{\rm a}^\eps\|_{ H^1(0,T;H^1)}\le C\|p_{\rm a}^\eps\|_{ H^1(0,T;H^{1/2}(\Gamma_{p,D}))},
\end{equation}
with $C$ independent on $\eps$. We take $q=\tilde p_{{\rm a}}^\eps-\check p_{{\rm a}}^\eps$ as test function in \eqref{eq:58} and we integrate on $(0,T)$ to obtain $\int_0^T\int_\Omega |\nabla^\eps\tilde p_{\rm a}^\eps|^2\mathop{{\rm d}x}\mathop{{\rm d}t}=\int_0^T\int_\Omega \nabla^\eps\check p_{\rm a}^\eps
\cdot\nabla^\eps\tilde p_{\rm a}^\eps\mathop{{\rm d}x}\mathop{{\rm d}t}$, whence the estimate
\begin{equation}\label{eq:82}
\|\nabla^\eps\tilde p_{\rm a}^\eps\|_{L^2(0,T;L^2)}\le \|\nabla^\eps\check p_{\rm a}^\eps\|_{L^2(0,T;L^2)}.
\end{equation} 
Next, we differentiate \eqref{eq:58} with respect to time and we take  $q=\partial_t(\tilde p_{{\rm a}}^\eps-\check p_{{\rm a}}^\eps)$ as test function. Then we integrate with respect to time to get $\int_0^T\int_\Omega |\nabla^\eps\partial_t\tilde p_{\rm a}^\eps|^2\mathop{{\rm d}x}\mathop{{\rm d}t}=\int_0^T\int_\Omega \nabla^\eps\partial_t\check p_{\rm a}^\eps
\cdot\nabla^\eps\partial_t\tilde p_{\rm a}^\eps\mathop{{\rm d}x}\mathop{{\rm d}t}$, which yields
\begin{equation}\label{eq:82b}
\|\nabla^\eps\partial_t\tilde p_{\rm a}^\eps\|_{L^2(0,T;L^2)}\le \|\nabla^\eps\partial_t\check p_{\rm a}^\eps\|_{L^2(0,T;L^2)}.
\end{equation} 
Putting together \eqref{eq:82} and \eqref{eq:82b} we have
\begin{equation}\label{eq:80}
\|\nabla^\eps\tilde p_{\rm a}^\eps\|_{H^1(0,T;L^2)}\le \|\nabla^\eps\check p_{\rm a}^\eps\|_{H^1(0,T;L^2)}, 
\end{equation}
 an estimate that is best written by decomposing the rescaled pressure gradient $\nabla^\eps\tilde p_{\rm a}^\eps$ into its transverse component $\eps^{-1}\partial_3\tilde p_{\rm a}^\eps$ and its plane component ${\overline\nabla}\tilde p_{\rm a}^\eps=(\partial_1\tilde p_{\rm a}^\eps,\partial_2\tilde p_{\rm a}^\eps)$:
\begin{equation}\label{eq:83} 
\|\partial_3\tilde p_{\rm a}^\eps\|^2_{H^1(0,T;L^2)}+\|{\eps\overline\nabla}\tilde p_{\rm a}^\eps\|_{H^1(0,T;L^2)}^2\le C \|\partial_3\check p_{\rm a}^\eps\|^2_{H^1(0,T;L^2)}+\|\eps{\overline\nabla}\check p_{\rm a}^\eps\|^2_{H^1(0,T;L^2)}.
\end{equation}
From \eqref{eq:81} and \eqref{eq:83} and from Assumption \eqref{eq:63d} we obtain 
\begin{equation}\label{eq:84}
\|\partial_3\tilde p_{\rm a}^\eps\|^2_{H^1(0,T;L^2)}+\|\eps{\overline\nabla}\tilde p_{\rm a}^\eps\|^2_{H^1(0,T;L^2)}\le C.
\end{equation}
\mynote{Now we write $\tilde p_{\rm a}^\eps(x)=p_{\rm a}(x_1,x_2,-h/2)+\int_{-h/2}^{x_3} \partial_3\tilde p_{\rm a}^\eps(x_1,x_2,z)\mathop{{\rm d}z}$ and we use \eqref{eq:84} to obtain
\[
\begin{aligned}
\int_\Omega|\tilde p_{\rm a}^\eps(t)|^2\mathop{{\rm d}x}
&=
\int_\Omega \Big|\tilde p_{\rm a}^\eps(x_1,x_2,-h/2,t)+\int_{-h/2}^{x_3} \partial_3\tilde p_{\rm a}^\eps(x_1,x_2,z,t)\mathop{{\rm d}z}\Big|^2\mathop{{\rm d}x}\\
&\le
2\int_\Omega \left|\tilde p_{\rm a}^\eps(x_1,x_2,-h/2,t)\right|^2\mathop{{\rm d}x}+
2\int_\Omega \Big|\int_{-h/2}^{x_3}\partial_3 \tilde p_{\rm a}^\eps(x_1,x_2,z,t)\mathop{{\rm d}z}\Big|^2\mathop{{\rm d}x}
\\
&\le
2h \|p_a(t)\|^2_{L^2(\Gamma_{p,D})}+\int_\Omega \Big(x_3+\frac h 2\Big)\int_{-h/2}^{x_3}|\partial_3\tilde p_{\rm a}^\eps(x_1,x_2,z,t)|^2\mathop{\mathrm{d}z}\mathop{{\rm d}x}\\
&\le
Ch\Big(\|p_{\rm a}(t)\|^2_{L^2(\Gamma_{p,D})}+h\|\tilde p_{\rm a}^\eps(t)\|^2_{L^2(\Gamma_{p,D})}\big).
\end{aligned}
\]
From the last inequality, on taking \eqref{eq:84} into account, we deduce
\begin{equation}\label{eq:47a}
\|\tilde p_{\rm a}^\eps\|^2_{L^2(0,T;L^2)}\le C.
\end{equation}}
From \eqref{eq:63d}, \eqref{eq:84}, and from Poincar\'e inequality we have
\begin{equation}\label{eq:47}
\|\tilde p_{\rm a}^\eps\|^2_{H^1(0,T;L^2)}\le C.
\end{equation}
From the bounds \eqref{eq:84} and \eqref{eq:47}  we conclude that there exists $\tilde p_{\rm a}\in H^1(0,T;L^2(\Omega))$ such that $\partial_3\tilde p_{\rm a}\in H^1(0,T;L^2(\Omega))$ and 
\begin{equation}\label{eq:57}
\tilde p_{\rm a}^\eps\rightharpoonup \tilde p_{\rm a}\quad\text{ and }\quad\partial_3\tilde p_{\rm a}^\eps\rightharpoonup \partial_3\tilde p_{\rm a}\quad\text{ in }H^1(0,T;L^2(\Omega)),
\end{equation}
for some subsequence. Moreover, by the continuity of the trace operator,
\begin{equation}\label{eq:64}
\tilde p_{\rm a}(x_1,x_2,\pm h/2,t)=p_{\rm a}^\pm(x_1,x_2,t)\quad\forall t\in[0,T).
\end{equation}
Next, from \eqref{eq:10h} and \eqref{eq:58} we have, for all $t\in[0,T)$,
\begin{equation}\label{eq:86}
\int_\Omega K^\eps_{33}\partial_3\tilde p^\eps_{\rm a}(t)\partial_3 q\dx+\eps^2\int_\Omega K^\eps_{\alpha\beta}\tilde p_{\rm a,\alpha}^\eps(t)\cdot q_{,\beta}\mathop{{\rm d}x}=0\quad\forall q\in H^1_{p,D}(\Omega).
\end{equation}
Thanks to \eqref{eq:84}, \eqref{eq:57}, and to assumption \eqref{eq:85}, we can pass to the limit in \eqref{eq:86} to obtain, for all $t\in[0,T)$,
\[
\int_\Omega  K_{33}\partial_3\tilde p_{\rm a}(t)\partial_3q\dx=0,\quad \forall q\in H^1_{p,D}(\Omega),
\] 
which implies, because of \eqref{eq:64}, that the representation \eqref{eq:87} holds. We conclude the proof by noting that since the limit $\tilde p_{\rm a}$ is uniquely determined the whole sequence converges.
\end{proof}
\medskip
\noindent We introduce the following space:
\[
\begin{aligned}
H^1_{KL}(\Omega;\mathbb R^3):=&\{\mathbf u\in H^1_{u,D}(\Omega;\mathbb R^3):E_{3i}\mathbf u=0\textrm{ a.e. in }\Omega\}.
\end{aligned}
\]
\begin{proposition}[Compactness]\label{prop:6}
There exist \blue $\bm u\in H^1(0,T;H^1(\Omega;\mathbb R^{3}))$ and $\eta\in H^1(0,T;L^2(\Omega))$ \black such that
\begin{subequations}\label{eq:71}
\begin{align}
\bm u^\eps&\rightharpoonup \bm u&&\text{in}\quad H^1(0,T;H^1(\Omega;\mathbb R^{3})),\\
E_{\alpha\beta}^\eps\bm u^\eps=E_{\alpha\beta}\bm u^\eps&\rightharpoonup E_{\alpha\beta}\bm u &&\text{in}\quad H^1(0,T;L^2(\Omega)),\\
E_{33}^\eps\bm u^\eps=\frac{\partial_3u_{3}^\eps}{\eps^2}&\rightharpoonup \eta &&\text{in}\quad H^1(0,T;L^2(\Omega)),\\
\color{black}\widehat p^\eps&\stackrel{*}\rightharpoonup 0 &&\text{in}\quad L^\infty(0,T;H^1(\Omega)),\label{eq:71d}
\end{align}
\end{subequations}
\blue for some subsequence. \black 
Moreover, the \blue limit \black $\bm u$ satisfies
\begin{subequations}
\begin{align}
\bm u(t)\in H^1_{KL}(\Omega;\mathbb R^3)
\end{align}
\end{subequations}
for a.e. $t\in(0,T)$.
\end{proposition}

\begin{proof}
Differentiate the first equation in \eqref{eq:9} and take  $\bm v=\partial_t\bfu^\eps$ to get, for a.e. $t\in(0,T)$,
\[
\begin{aligned}
&\int_{\Omega}\big(\mathbb C^\eps\partial_t\bm E^\eps\bfu^\eps(t)\cdot\partial_t\bm E^\eps\bfu^\eps(t)- \partial_t\color{black}\widehat p^\eps(t)\black\partial_t{\mathop{\rm div}}^\eps\bfu^\eps(t)\mathop{{\rm d}x}
\\
&\hskip 6em  =\int_{\Omega}\big(\partial_t\blue\widehat{\bm f}^\eps(t)\black\cdot\partial_t\bfu^\eps(t)+\blue\partial_t\widetilde p_{\rm a,\eps}(t)\partial_t\operatorname{div}^\eps\bm u^\eps(t)\big)\mathop{{\rm d}x}.
\end{aligned}
\]
Choose $\partial_t \blue \widehat p^\eps(t)$ as test in the second equation in \eqref{eq:9} to obtain
\[
  \int_{\Omega}\big( {\rm div}^\eps\partial_t\bfu^\eps(t)\partial_t \blue \widehat p^\eps(t)+\bm K^\eps\nabla^\eps \color{black}\widehat p^\eps(t)\cdot\nabla^\eps \partial_t \blue \widehat p^\eps(t)\black\big)\mathop{{\rm d}x}=0.
\]
Adding the above equations yields
\[ 
\begin{aligned}
\int_{\Omega}\big(\mathbb C^\eps\partial_t\bm E^\eps\bfu^\eps(t)\cdot\partial_t\bm E\bfu^\eps(t)&+\mb K^\eps\nabla^\eps \color{black}\widehat p^\eps(t)\cdot\partial_t\nabla^\eps\color{black}\widehat p^\eps(t)\big)\dx
\\
&=\int_{\Omega}\big(\partial_t\bar{\bm f}^\eps(t)\cdot\partial_t\bfu^\eps(t)+\blue\partial_t\widetilde p_{\rm a,\eps}(t)\partial_t\operatorname{div}^\eps\bm u^\eps(t)\big)\mathop{{\rm d}x}.
\end{aligned}
\]
which holds for a.e. $t\in(0,T)$. From the above equation we obtain the following estimate
\begin{equation}\label{eq:72}
\|\bm E^\eps\bm u^\eps\|^2_{H^1(0,T;L^2)}+\|\nabla^\eps\color{black}\widehat p^\eps\|^2_{L^\infty(0,T;L^2)}\le C,
\end{equation}
whence (\ref{eq:71}c). Using \eqref{eq:72}, Korn's inequality, Poincar\'e inequality, and the boundary conditions we obtain that 
\begin{equation*}
\begin{aligned}
&\|\bm u^\eps\|_{H^1(0,T;H^1)}\le C,\\
&\|\color{black}\widehat p^\eps\|_{L^\infty(0,T;H^1)}\le C.
\end{aligned}
\end{equation*}
These inqualities imply (\ref{eq:71}a,b) and
\[
\color{black}\widehat p^\eps\stackrel{*}\rightharpoonup \color{black}\widehat p\quad\text{in}\quad L^\infty(0,T;H^1(\Omega)),
\]
with $\color{black}\widehat p=0$ on $\Gamma_{p,D}$. Since $\|\partial_3\color{black}\widehat p^\eps\|_{L^\infty(0,T;L^2)}\le C\varepsilon$ by \eqref{eq:72}, we conclude that $\partial_3\color{black}\widehat p=0$ and hence $\color{black}\widehat p=0$.
\end{proof}
As a preliminary step, we need the following result.
\begin{proposition}[Characterization of $\eta$]\label{prop:7}
The limit $\eta$ is given by
\begin{equation}\label{eq:49}
  \eta=\frac{\tilde p_{\rm a}-\mathbb C_{33\alpha\beta}E_{\alpha\beta}\bm u}{\mathbb C_{3333}}.
\end{equation}
\end{proposition}
\begin{proof}
In the first equation of \color{black}\eqref{eq:9} \color{black}we take as test function $\bm v=\eps^2 \varphi\bm e_3$ with $\varphi(x)=\int_0^{x_3}\psi(x_1,x_2,z){\rm d}z$ where $\psi\in C^\infty_0(\Omega)$. On letting $\eps$ tend to 0, we obtain, for almost all times,
\begin{equation*}
  \int_\Omega \left(\mathbb C_{3333}\eta+\mathbb C_{33\gamma\delta}E_{\gamma\delta}\bm u-\tilde p_{\rm a}\right)\psi\,{\rm d}x=0.
\end{equation*}
By the arbitrariness of $\psi$ we obtain \eqref{eq:49}.
\end{proof}
\begin{remark}
In the special case when $\mathbb C$ is isotropic, we have
\begin{equation}\label{eq:922}
 \mathbb C\bm E=2G\bm E+(K-\frac 23 G)\textrm{tr}(\bm E)\bfI,  
\end{equation}
and \eqref{eq:49} becomes
\[
E_{33}=\frac{\tilde p_{\rm a}-\left(K-\frac 23 G\right)E_{\gamma\gamma}}{\frac 4 3 G+K}
\]
\end{remark}

\begin{theorem}\label{thm:2}
Let $\bm u^\eps$ and $\color{black}\widehat p^\eps$ the solution of \eqref{eq:9}. As $\varepsilon$ tends to $0$, 
\begin{subequations}
\begin{align}
\bm u^\eps&\rightharpoonup \bm u&&\text{in}\quad H^1(0,T;H^1(\Omega;\mathbb R^{3})),\\
\color{black}\widehat p^\eps&\stackrel{*}\rightharpoonup 0 &&\text{in}\quad L^\infty(0,T;H^1(\Omega)).
\end{align}
\end{subequations}
Moreover, the weak limit $\bm u$ satisfies, for all $t\in[0,T)$, the following variational equation
\begin{equation}\label{eq:88}
\left\{
\begin{aligned}
&\bm u\in H^1(0,T;H^1_{KL}(\Omega)),\\
&\int_\Omega \left(\bar{\mathbb C}_{\alpha\beta\gamma\delta}E_{\alpha\beta}\bm u(t)-D_{\gamma\delta}\tilde p_{\rm a}(t)\right)E_{\gamma\delta}{\bm v}{\rm d}x=\int_\Omega \bm f(t)\cdot{\bm v}{\rm d}x\\
&\qquad\forall{\bm v}\in H^1_{KL}(\Omega;\mathbb R^3).
\end{aligned}
\right.
\end{equation}
where
\begin{equation}\label{eq:98}
  \overline{\mathbb C}_{\alpha\beta\gamma\delta}=\mathbb C_{\alpha\beta\gamma\delta}-\frac{\mathbb C_{\alpha\beta 33}
\mathbb C_{33\gamma\delta}}{\mathbb C_{3333}},\quad D_{\alpha\beta}=\delta_{\alpha\beta}\color{black}-\color{black}\frac{\mathbb C_{33\alpha\beta}}{\mathbb C_{3333}}.
\end{equation}
\end{theorem}
\begin{proof}
In the weak formulation \eqref{eq:9} we take $\bm v(x,t)=\varphi(t) \bar{\bm v}(x)$, with $\bar{\bm v}\in H^1_{KL}(\Omega;\mathbb R^3)$.
We have $E_{3 i}\bar{\bm v}=0$, and hence 
$$
\mathbb C^\eps\bm E^\eps\bm u^\eps\cdot\bm E^\eps\bar{\bm v}=\mathbb C_{ij\gamma\delta}^\eps E^\eps_{ij}\bm u^\eps E_{\gamma\delta}\bar{\bm v}=
\mathbb C_{\alpha\beta\gamma\delta}^\eps E^\eps_{\alpha\beta}\bm u^\eps E_{\gamma\delta}\bar{\bm v}+\mathbb C_{33\gamma\delta}^\eps E^\eps_{33}\bm u^\eps E_{\gamma\delta}\bar{\bm v},
$$
where the second equality follows from \eqref{eq:10e}. Hence, the first equation in \eqref{eq:9} takes the form
\begin{equation}\label{eq:73}
\begin{aligned}
\int_0^T\varphi(t)&\int_{\Omega}\big(\mathbb C_{\alpha\beta\gamma\delta}^\eps E^\eps_{\alpha\beta}\bm u^\eps(t) E_{\gamma\delta}\bar{\bm v}+\mathbb C_{33\gamma\delta}^\eps E^\eps_{33}\bm u^\eps(t) E_{\gamma\delta}\bar{\bm v}- \color{black}\widehat p^\eps(t)\bar v_{\alpha,\alpha})\mathop{{\rm d}x}{\rm d}t\\
&=\int_0^T\varphi(t)\int_{\Omega}\big({\bm f}^\eps(t)\cdot\bar{\bm v}+\tilde p_{\rm a}^\eps(t)\partial_\alpha\bar v_{\alpha}\big)\mathop{{\rm d}x}{\rm d}t.
\end{aligned}
\end{equation}
Thus, thanks to Lemma \ref{lem:3}, Proposition \ref{prop:6}, and to assumption \eqref{eq:63a} we can pass to the limit in \eqref{eq:73} to get
\[
\begin{aligned}
\int_0^T\varphi(t)&\int_\Omega \left(\mathbb C_{\alpha\beta\gamma\delta}E_{\alpha\beta}\bm u(t)E_{\gamma\delta}\bar{\bm v}+\mathbb C_{33\gamma\delta}\,\eta(t)E_{\gamma\delta}\bar{\bm v}\right){\rm d}x{\rm d}t\\
&=\int_0^T\varphi(t)\int_\Omega \left(\bm f(t)\cdot\bar{\bm v}+\tilde p_{\rm a}(t)\bar v_{\alpha,\alpha}\right){\rm d}x{\rm d}t.
\end{aligned}
\]
By the arbitrariness of $\varphi$, we have 
\[
\int_\Omega \left(\mathbb C_{\alpha\beta\gamma\delta}E_{\alpha\beta}\bm u(t)E_{\gamma\delta}\bar{\bm v}+\mathbb C_{33\gamma\delta}\,\eta(t)E_{\gamma\delta}\bar{\bm v}-\tilde p_{\rm a}(t)\bar v_{\alpha,\alpha}\right){\rm d}x=\int_\Omega\bm f(t)\cdot\bar{\bm v}{\rm d}x,
\]
which holds for all $t\in[0,T)$ by continuity. Now we use Proposition \ref{prop:7} to express $\eta$ in terms of $E_{\alpha\beta}\bm u$ and $\tilde p_{\rm a}$ to obtain the thesis.
\end{proof}


\paragraph{The plate equations.}
As a first step towards the deduction of plate equations is the following representation result, whose proof may be found, for instance, in \cite[Thm. 1.4-1.(c)]{Ciarl1997}).
\begin{proposition}[Characterization of $H^1_{KL}(\Omega;\mathbb R^3)$]\label{prop:5}
A displacement $\bm v:\Omega\to\mathbb R^3$ is an element of $H^1_{KL}(\Omega;\mathbb R^3)$ if and only if there exist
$z_\alpha\in H^1(\omega)$ and $z_3\in H^2(\omega)$ such that $z_i=0$ on $\gamma_{u,D}$, $\partial_nz_3=0$ on $\gamma_{u,D}$ and
\begin{equation}\label{eq:66}
v_\alpha(x)=z_\alpha(x_1,x_2)-x_3 z_{3,\alpha}(x_1,x_2),\qquad v_3(x_1,x_2)=z_3(x_1,x_2)
\end{equation}
for a.e. $x\in\Omega$.
\end{proposition}
According to Proposition \ref{prop:5} there exist 
\begin{equation}\label{eq:46}
w_\alpha \in H^1(0,T;H^1(\omega)),\qquad w_3\in H^1(0,T;H^2(\omega))
\end{equation} 
such that the limit displacement field obtained in Theorem \ref{thm:2} admits the representation 
\begin{equation}\label{eq:67}
  u_\alpha(x,t)=w_\alpha(x_1,x_2,t)-x_3\partial_\alpha w_3(x_1,x_2),\qquad u_3(x,t)=w_3(x_1,x_2,t)
\end{equation}
for a.e. $(x,t)\in\Omega\times(0,T)$. Thus,
\begin{equation}\label{eq:95}
  E_{\alpha\beta}\bm u=E_{\alpha\beta}\bm w-x_3\partial_{\alpha\beta}w_3.
\end{equation}
As a second step, we turn to the variational statement in \eqref{eq:88} that characterizes the limit $\bm u$, namely,
\begin{equation}\label{eq:90}
\begin{aligned}
&\int_\Omega S_{\alpha\beta}(t)E_{\gamma\delta}{\bm v}{\rm d}x=\int_\Omega \bm f(t)\cdot{\bm v}{\rm d}x\\
&\quad\forall \bm v\in H^1_{KL}(\Omega):\bm v=\bm 0\text{ on }\Gamma_{\rm u,D},
\end{aligned}
\end{equation}
with
\begin{equation}\label{eq:96}
S_{\alpha\beta}:=\overline{\mathbb C}_{\alpha\beta\gamma\delta}E_{\gamma\delta}\bm u-D_{\alpha\beta}\tilde p_{\rm a},
\end{equation}
a statement that holds at all times $t\in[0,T)$. We substitute into \eqref{eq:90} the representation \eqref{eq:66} and  \eqref{eq:67} of the unknown $\bm u$ and the test $\bm v$. The result is a pair of variational statements involving the unknowns $\bm w$ and the test functions $\bm z$. The first statement involves the in--plane components $z_\alpha$ of the test function $\bm z$, and has the form
\begin{equation}
\begin{aligned}
 &\int_\omega N_{\alpha\beta}(t)E_{\alpha\beta}\bm{z}\,{\rm d}x_1{\rm d}x_2=\int_\omega r_\alpha(t) z_\alpha{\rm d}x_1{\rm d}x_2\\
&\quad\forall \bm z=(z_1,z_2)\in H^1(\omega;\mathbb R^2): \bm z=\bm 0 \text{ on }\gamma_{u,D},
\end{aligned}
\end{equation}
with
\begin{equation}\label{eq:91}
N_{\alpha\beta}(x_1,x_2,t):=\int_{-h/2}^{+h/2}S_{\alpha\beta}(x,t){\rm d} x_3,
\end{equation}
and
\begin{equation}
r_\alpha(x_1,x_2,t):=\int_{-h/2}^{+h/2}f_\alpha(x,t){\rm d}x_3;
\end{equation}
the second statement is deduced by considering the transversal component $z_3$ of the test $\bm z$, and reads 
\begin{equation}
\begin{aligned}
  &\int_\omega M_{\alpha\beta}(t)(-\partial_{\alpha\beta}z_3){\rm d}x_1{\rm d}x_2=\int_\omega s_\alpha(t)(-\partial_\alpha z_3){\rm d}x_1{\rm d}x_2\\
&\qquad\forall z_3\in H^2(\omega): z_3=0\text{ and }\partial_nz_3=0\text{ on }\gamma_{u,D},
\end{aligned}
\end{equation}
where
\begin{equation}\label{eq:97}
M_{\alpha\beta}(x_1,x_2,t):=\int_{-h/2}^{+h/2}x_3 S_{\alpha\beta}(x,t){\rm d} x_3,
\end{equation}
and
\begin{equation}
s_\alpha(x_1,x_2,t):=\int_{-h/2}^{+h/2}x_3 f_\alpha(x,t){\rm d}x_3.
\end{equation}
On introducing the shorthand notation:
\begin{equation}
  \varphi^{(i)}(x_1,x_2,t)=\frac{(i+1)2^i}h\int_{-h/2}^{h/2}\left(\frac{x_3}{h}\right)^i\varphi(x,t){\rm d}x_3.
\end{equation}
and on combining \eqref{eq:95}  and \eqref{eq:96}, and \eqref{eq:87} with \eqref{eq:91} we can write the explicit expression of the tension forces $N_{\alpha\beta}$ as
\begin{equation}\label{eq:93}
N_{\alpha\beta}=h\overline{\mathbb C}_{\alpha\beta\gamma\delta}^{(0)}E_{\gamma\delta}\bm w+\frac {h^2}4 \overline{\mathbb C}^{(1)}_{\alpha\beta\gamma\delta}(-\partial_{\gamma\delta}w_3)-h D_{\alpha\beta}^{(0)}\frac{p^++p^-}{2}-(\zeta D_{\alpha\beta})^{(0)}(p^+-p^-);
\end{equation}
likewise, combination with \eqref{eq:97} yields
\begin{equation}\label{eq:94}
  M_{\alpha\beta}=\frac {h^3}{12}\overline{\mathbb C}_{\alpha\beta\gamma\delta}^{(2)}(-\partial_{\gamma\delta}w_3)+\frac {h^2}4\overline{\mathbb C}^{(1)}_{\alpha\beta\gamma\delta}E_{\gamma\delta}\bm w-\frac {h^2}{4}D_{\alpha\beta}^{(1)}\frac{p^++p^-}2-\frac {h}{4}(\zeta D_{\alpha\beta})^{(1)}(p^+-p^-).
\end{equation}

\paragraph{Isotropic material response independent on $x_3$.}
We now specialize our results to a situation when the tensors $\mathbb C$ and $\bm K$ are isotropic and do not depend on the coordinate $x_3$. In this case, we have $\mathbb C_{ijkl}=2G\delta_{ik}\delta_{jl}+\lambda\delta_{ij}\delta_{kl}$ and $K_{ij}=\kappa\delta_{ij}$ with $\kappa$ a positive constant. An elementary calculation yields
\begin{equation}
  \overline{\mathbb C}_{\alpha\beta\gamma\delta}=2G\delta_{\alpha\gamma}\delta_{\beta\delta}+\frac {2G} {2G+\lambda}\lambda\delta_{\alpha\beta}\delta_{\gamma\delta},\qquad D_{\alpha\beta}=\delta_{\alpha\beta}.
\end{equation}
In this case, we have $\overline{\mathbb C}^{(i)}_{\alpha\beta\gamma\delta}=\overline{\mathbb C}_{\alpha\beta\gamma\delta}$ if $i$ is even and $\overline{\mathbb C}^{(i)}_{\alpha\beta\gamma\delta}=0$ if $i$ is odd. Likewise, $D_{\alpha\beta}^{(i)}=\delta_{\alpha\beta}^{(i)}=\delta_{\alpha\beta}$ if $i$ is even and  $D_{\alpha\beta}^{(i)}=0$ if $i$ is odd. Moreover, as observed in Remark \ref{rem:4}, the independence of $\bm K$ on $x_3$ entails that $\zeta(x)=x_3$ in the representation formula \eqref{eq:87}; thus, $(\zeta D_{\alpha\beta})^{(i)}=(x_3 \delta_{\alpha\beta})^{(i)}=h(i+1)/(2(i+2))\delta_{\alpha\beta}$ if $i$ is odd and  $(\zeta D_{\alpha\beta})^{(i)}=0$ if $i$ is even. As a result, \eqref{eq:93} and \eqref{eq:94} become, respectively,
\begin{equation}
  N_{\alpha\beta}=2Gh E_{\alpha\beta}\bm w+\frac {2G\lambda h} {2G+\lambda} (\partial_\gamma w_\gamma) \delta_{\alpha\beta}-h\frac{p^++p^-}{2}\delta_{\alpha\beta},
\end{equation}
and
\begin{equation}
  M_{\alpha\beta}=\frac {h^3}{12}2G(-\partial_{\alpha\beta} w_3)-\frac {h^3}{12}\frac {2G\lambda} {2G+\lambda} \Delta w_3 \delta_{\alpha\beta}-\frac {h^2}{12}(p^+-p^-)\delta_{\alpha\beta}.
\end{equation}

\color{black}

\section{A diffusive plate model}\label{sec:diff-plate-model}

In this section we consider a plate with mobility through the thickness of the plate much smaller than that in the plane, that is 
$K_{33}^\eps\ll K_{\alpha\beta}^\eps$.  A similar assumption was also made in \cite{LucanN2012IJSS}. In particular, we assume that $\bm K^\eps$ has the following scaling
$$
\bm K^\eps:=
\left(
\begin{array}{cc}
K_{\alpha\beta} & 0\\
0 & \eps^2 K_{33}
\end{array}
\right)
$$
where for simplicity we have taken $K_{\alpha 3}^\eps=0$ and where $K_{\alpha\beta}$ and $K_{33}$ are functions whose regularity is specified below.
Under this position, for any functions $f$ and $g$ we have that
$
\bm K^\eps \nabla^\eps f\cdot \nabla^\eps g= \bm K \nabla f\cdot \nabla g
$
where
$$
\bm K:=
\left(
\begin{array}{cc}
K_{\alpha\beta} & 0\\
0 &  K_{33}
\end{array}
\right).
$$
Within this setting, Problem \eqref{eq:58} rewrites as
\begin{equation}\label{eq:758}
\forall t\in[0,T]:\begin{cases}
&\tilde p_{\rm a}^\eps\in H^1(\Omega), \qquad \tilde p^\eps_{\rm a}(t)=p_{\rm a}^\eps(t)\text{ on }\Gamma_{p,D},\\[0.5em]
&\displaystyle\int_\Omega \bm K\nabla \tilde p_{\rm a}^\eps(t)\cdot\nabla q\mathop{{\rm d}x}=0
\quad\forall q\in H^1_{p,D}(\Omega),
\end{cases}
\end{equation}
and  Problem \eqref{eq:9} as
\begin{equation}\label{eq:79}
\left\{ 
\begin{aligned}
  &\bfu^\varepsilon\in H^1(0,T;H^1_{u,D}(\Omega;\mathbb R^3)),\\
&\bar p^\varepsilon\in L^\infty(0,T;H^1_{p,D}(\Omega))\cap H^1(0,T;L^2(\Omega))\quad\text{ with } \quad \bar p^\eps(0)=\color{black}\bar p^\eps_{0},\\
&\int_{\Omega}\big(\mathbb C^\eps\bm E^\eps\bfu^\eps(t)\cdot\bm E^\eps\bm v- \bar p^\eps(t){\mathop{\rm div}}^\eps\bm v)\mathop{{\rm d}x}=\int_{\Omega}\big({\bm f}^\eps(t)\cdot\bm v+\tilde p_{\rm a}^\eps{\mathop{{\rm div}^\eps}}\bm v\big)\mathop{{\rm d}x}\\
&\qquad\forall\bm v\in H^1_{u,D}(\Omega;\mathbb R^3),\color{black}\forall t\in[0,T],\\
&\int_{\Omega} {\rm div}^\eps\dot\bfu^\eps(t)q+\bm K\nabla \bar p^\eps(t)\cdot\nabla q\mathop{{\rm d}x}=0\qquad\forall q\in H^1_{p,D}(\Omega),\text{ for a.e. }t\in(0,T).\\
\end{aligned}\right.
\end{equation} 
We keep all the assumption stated in the previous section except \eqref{eq:10h}, \eqref{eq:10i},
and  \eqref{eq:85}, that we replace by
\begin{subequations}\label{eq:10}
\begin{align}
&\mathbf K\in L^\infty(\Omega;\mathbb R^{3\times 3}_{\rm Sym}),\\
&\text{$\exists c_{\bm K}>0:\mathbf K\bm a\cdot\bm a\ge c_{\bm K}|\bm a|^2\quad\forall \bm a\in\mathbb R^3$}\quad \text{a.e. in }\Omega.
\end{align}
\end{subequations}

The asymptotic behaviour of problems \eqref{eq:758} and \eqref{eq:79} can be easily studied by following the same steps taken in the previous section. In particular, Lemma  \ref{lem:3} will be replaced by the following lemma.

\begin{lemma}\label{lem:73}
Let $\tilde p_{\rm a}^\eps(t)\in H^1(\Omega)$ solve \eqref{eq:758} for all $t\in[0,T]$. Then
\[
\tilde p_{\rm a}^\eps\rightharpoonup \tilde p_{\rm a}\quad\textrm{in }H^1(0,T;H^1(\Omega)),
\]
where $\tilde p_{\rm a}$ is the solution of
\begin{equation}\label{eq:7758}
\forall t\in[0,T]:\begin{cases}
&\tilde p_{\rm a}\in H^1(\Omega), \qquad \tilde p_{\rm a}(t)=p_{\rm a}(t)\text{ on }\Gamma_{p,D},\\[0.5em]
&\displaystyle\int_\Omega \bm K\nabla \tilde p_{\rm a}(t)\cdot\nabla q\mathop{{\rm d}x}=0
\quad\forall q\in H^1_{p,D}(\Omega).
\end{cases}
\end{equation}
\end{lemma}
The proof of the Lemma \ref{lem:73} coincides with that of Lemma \ref{lem:3}.
Proposition \ref{prop:6} still holds but with \eqref{eq:71d} replaced by
$$
\bar p^\eps\stackrel{*}\rightharpoonup \bar p\quad\textrm{in }L^\infty(0,T;H^1(\Omega))
$$
where $\bar p \in L^\infty(0,T;H^1_{p,D}(\Omega))$. Within this setting $\bar p\ne 0$, contrary to the previous section, because only $\nabla\bar p^\eps$ ,and not $\nabla^\eps\bar p^\eps$,  it is bounded in ${L^\infty(0,T;L^2)}$. If in the proof of  Proposition \ref{prop:7} we take into account that $\bar p\ne 0$
we find 
\begin{equation}\label{eq:749}
  \eta=\frac{\tilde p_{\rm a}+\bar p-\mathbb C_{33\alpha\beta}E_{\alpha\beta}\bm u}{\mathbb C_{3333}}.
\end{equation}
Note that since $\eta, E_{\alpha\beta}\bm u\in H^1(0,T;L^2(\Omega))$ it follows that also $\bar p\in H^1(0,T;L^2(\Omega))$.

\begin{theorem}\label{thm:72}
Let $\bm u^\eps$ and $\bar p^\eps$ the solution of \eqref{eq:79}. As $\varepsilon$ tends to $0$, 
\begin{subequations}
\begin{align}
\bm u^\eps&\rightharpoonup \bm u&&\text{in}\quad H^1(0,T;H^1(\Omega;\mathbb R^{3})),\\
\bar p^\eps&\stackrel{*}\rightharpoonup \bar p &&\text{in}\quad L^\infty(0,T;H^1(\Omega)).
\end{align}
\end{subequations}
Moreover, the weak limits $\bm u$ and $\bar p$ satisfy, for all $t\in[0,T]$, the following variational equation
\begin{equation}\label{eq:788}
\left\{
\begin{aligned}
&\bm u\in H^1(0,T;H^1_{KL}(\Omega)),\\
&\bar p\in L^\infty(0,T;H^1_{p,D}(\Omega))\cap H^1(0,T;L^2(\Omega))\quad\text{ with } \quad \bar p(0)=\color{black}\bar p_{0},\\
&\int_\Omega \left(\bar{\mathbb C}_{\alpha\beta\gamma\delta}E_{\alpha\beta}\bm u(t)-D_{\gamma\delta}(\tilde p_{\rm a}+\bar p)(t)\right)E_{\gamma\delta}{\bm v}{\rm d}x=\int_\Omega \bm f(t)\cdot{\bm v}{\rm d}x\\
&\qquad\forall{\bm v}\in H^1_{KL}(\Omega;\mathbb R^3).\\
&\int_{\Omega} \left(D_{\alpha\beta} E_{\alpha \beta}  \dot\bfu(t) - \frac{\partial_t \tilde p_{\rm a}(t)+\partial_t \bar p(t)}{\mathbb C_{3333}}\right)q+\bm K\nabla \bar p(t)\cdot\nabla q\mathop{{\rm d}x}=0\\
&\qquad\forall q\in H^1_{p,D}(\Omega),\text{ for a.e. }t\in(0,T),
\end{aligned}
\right.
\end{equation}
where $\overline{\mathbb C}$ and $D$ are defined in \eqref{eq:98}.
\end{theorem}

Again, the proof of Theorem \ref{thm:72} is very similar to that of Theorem \ref{thm:2}; essentially it suffices to recall that  $\bar p\ne 0$ and that by Proposition \ref{prop:6} and \eqref{eq:749} we have that 
$$
{\rm div}^\eps\dot\bfu^\eps = {\rm tr} E^\eps \dot\bfu^\eps \rightharpoonup E_{\alpha \alpha}  \dot\bfu+\dot \eta= D_{\alpha\beta} E_{\alpha \beta}  \dot\bfu - \frac{\partial_t \tilde p_{\rm a}+\partial_t \bar p}{\mathbb C_{3333}},
$$
where we used \eqref{eq:749}.

We decided to write Theorem \ref{thm:72} with the same notation used in the previous section even if we could have written it, in a more compact form, in terms of 
$$
p(t):=\tilde p_{\rm a}(t)+\bar p(t).
$$

We finally note that problem \eqref{eq:788} is not a two dimensional problem. Indeed, while the  balance equation can be written over the two dimensional domain $\omega$, since the test function ${\bm v}$ is a Kirchhoff-Love type of displacement, the last equation appearing in \eqref{eq:788} cannot since the diffusion is throughout $\Omega$.

\bibliographystyle{abbrv}

\end{sloppypar}
\end{document}